\documentclass[lettersize,journal]{IEEEtran}
\IEEEoverridecommandlockouts
\usepackage{cite}
\usepackage{amsmath,amssymb,amsfonts,amsthm}
\usepackage{algorithmic}
\usepackage{graphicx}
\usepackage{textcomp}
\usepackage{algorithm}
\usepackage{multirow}
\usepackage{caption}
\usepackage{xcolor,soul}
\usepackage{bm}
\usepackage{dsfont}
\usepackage{siunitx}
\usepackage{enumitem}
\newtheorem{remark}{Remark}
\newtheorem{definition}{Definition}
\def\BibTeX{{\rm B\kern-.05em{\sc i\kern-.025em b}\kern-.08em
    T\kern-.1667em\lower.7ex\hbox{E}\kern-.125emX}}
\newtheorem{theorem}{Theorem}

\newtheorem{assumption}{Assumption}
\newtheorem{lemma}{Lemma}
\newtheorem{proposition}{Proposition}

\usepackage[hidelinks,colorlinks,linkcolor=blue,filecolor=blue,citecolor=blue,urlcolor=blue]{hyperref}

\begin{document}

\title{Computation and Critical Transitions  of Rate-Distortion-Perception Functions  With Wasserstein Barycenter}

\author{Chunhui~Chen,
        Xueyan~Niu,
        Wenhao~Ye,
        Hao~Wu,
        and~Bo~Bai,
\thanks{Chunhui Chen, Wenhao Ye, and Hao Wu are with the Department of Mathematical Sciences, Tsinghua University, Beijing, China (e-mail: \{cch21, yewh20\}@mails.tsinghua.edu.cn; hwu@tsinghua.edu.cn).}
\thanks{Xueyan Niu and Bo Bai are with Theory Lab, 2012 Labs, Huawei Technologies Co., Ltd. (e-mail: \{niuxueyan3, baibo8\}@huawei.com).}
\thanks{\textit{Corresponding author: Xueyan Niu.}}
\thanks{This paper was presented in part at the 2023 IEEE International Symposium on Information Theory.}
}

\maketitle

\begin{abstract}
The information rate-distortion-perception (RDP) function characterizes the three-way trade-off between description rate, average distortion, and perceptual quality measured by discrepancy between probability distributions and has been applied to emerging areas in communications empowered by generative modeling. 
We study several variants of the RDP functions through the lens of optimal transport to characterize their critical transitions.
By transforming the information RDP function into a Wasserstein Barycenter problem, we identify the critical transitions when one of the constraints becomes inactive.
Further, the non-strictly convexity brought by the perceptual constraint can be regularized by an entropy regularization term.
We prove that the entropy regularized model converges to the original problem and propose an alternating iteration method based on the Sinkhorn algorithm to numerically solve the regularized optimization problem. 
In many practical scenarios, the computation of the Distortion-Rate-Perception (DRP) function offers a solution to minimize distortion and perceptual discrepancy under rate constraints. However, the interchange of the rate objective and the distortion constraint significantly amplifies the complexity. The proposed method effectively addresses this complexity, providing an efficient solution for DRP functions. Using our numerical method, we propose a reverse data hiding scheme that imperceptibly embeds a secret message into an image, ensuring perceptual fidelity and achieving a significant improvement in the perceptual quality of the stego image compared to traditional methods under the same embedding rate. Our theoretical results and numerical method lay an attractive foundation for steganographic communications with perceptual quality constraints.

\end{abstract}


\IEEEpeerreviewmaketitle

\section{Introduction}

Lossy compression plays a vital role in the communication and storage of multimedia data. 
As the cornerstone of lossy compression,
the classical Rate-Distortion (RD) theory \cite{shannon1959coding} studies the trade-off between the bit rate used for representing data and the distortion caused by compression \cite{DBLP:books/wi/01/CT2001}. 
The reconstruction quality is traditionally measured by a per-letter distortion metric, such as the mean-squared error (MSE). However, it has been shown that in practice, minimizing the distortion does not result in perceptually satisfying output for human subjects \cite{DBLP:journals/corr/SanturkarBS17}. Since high perceptual quality may come at the expense of distortion \cite{DBLP:conf/iccv/AgustssonTMTG19,DBLP:conf/cvpr/BlauM18,freirich2023the}, researchers are motivated to extend the RD theory by bringing perception into account \cite{blau2019rethinking}, \cite{zhang2021universal}.
Blau and Michaeli proposed the information Rate-Distortion-Perception (RDP) functions in \cite{blau2019rethinking}, which translate the established machine learning practice to an information-theoretic problem.

The introduction of the perceptual quality constraint brings new challenges to the classical rate-distortion problem. Although several coding theorems under a variety of perceptual measures and different system settings \cite{matsumoto2018introducing, theis2021coding, wagner2022rate, chen2022rate, niu2023conditional, hamdi2023rate, salehkalaibar2024rate} have been derived,
analytically characterizing the interaction between the distortion and the perceptual constraints and closed form expressions to the RDP problem are often intractable. As witnessed in both theoretical studies \cite{zhang2021universal,wagner2022rate,niu2023conditional,hamdi2024the,salehkalaibar2024rate} and empirical results \cite{blau2019rethinking,zhang2021universal}, the rate regions for the RDP functions are piecewise smooth, which signifies critical transitions across domains where the constraints degenerate.  It is therefore compelling to characterize these behaviors to affirm the observed transitions and inform future perception-oriented compression and communications.
   Traditionally, the Blahut–Arimoto (BA) algorithm \cite{DBLP:journals/tit/Arimoto72,DBLP:journals/tit/Blahut72} has been successful in the computation of capacities and RD functions. Recently, a generalized BA type algorithm \cite{10206646} has been proposed to approximately solve the RDP function with a perception constraint that belongs to the family of $f$-divergences. However, to the best of our knowledge, we have not seen any generalization of the BA algorithm to compute RDP functions with another crucial family of perception constraints -- the Wasserstein metric. 
The first difficulty may be that RDP functions have more independent variables than RD functions, while in each alternating iteration step in the BA algorithm, all variables, except the updating one, need to be fixed. 
Secondly, Wasserstein metric does not possess an explicit derivative in terms of corresponding variables. Therefore, it cannot be directly expressed as an iterative expression like those in the family of $f$-divergences. Several works focus on computing RDP functions using data-driven methods, such as generative adversarial networks, which minimize a weighted combination of the distortion and the perception \cite{blau2019rethinking,zhang2021universal}.
However, these learning-based methods are often computationally demanding and lack interpretability.

We approach the RDP problem from the perspective of Optimal Transport (OT). We introduce the Wasserstein Barycenter model for Rate-Distortion-Perception functions, termed WBM-RDP, along with corresponding algorithms. Utilizing this new optimization framework, we uncover a close connection between the RDP problem and the Wasserstein Barycenter problem. Referring to the approach in a recent paper \cite{wu2022communication} of RD, we reformulate RDP functions in an OT form with an additional constraint on perceptual fidelity.
However, compared to RD functions, the introduction of an additional perceptual constraint destroys the origin simplex structure in the RDP functions. To handle this issue, we prove that the additional constraint can be equivalently converted to a set of linear constraints by introducing an auxiliary variable. 
Consequently, we observe that the new model appears to be in the form of the celebrated Wasserstein Barycenter problem \cite{pass2015multi,DBLP:conf/icml/CuturiD14,DBLP:journals/siamma/AguehC11}, as it can be viewed as a minimizer over two couplings (i.e., the transition mappings and the newly introduced variable) between Barycenter (i.e., the reconstruction distribution) and the source distribution. Then, the objective of the optimization is to compute a weighted distance according to the Wasserstein metric. Specifically, we determine the boundaries at which these perceptual constraints become inactive, causing the RDP function to degenerate into the RD function. This allows us to characterize the interplay between various distortion and perception constraints that are inherent to the information-theoretic RDP functions (see Fig.~\ref{curve}).
In addition to the rate function, the equivalent OT formulation also enables us to study the corresponding distortion and perception functions in relation to the rate, akin to the classic distortion-rate function.

With the reformulated Wasserstein Barycenter problem, we design an algorithm to tackle the WBM-RDP. 
One challenge is that the WBM-RDP is not strictly convex. We construct an entropy-regularized formulation of the WBM-RDP and show that the new form admits a unique optimal solution, and that it converges to the origin WBM-RDP.
After obtaining the Lagrangian of the entropy regularized WBM-RDP, we observe that the degrees of freedom therein can be divided into three groups to be optimized alternatively with closed-form solutions. 
As such, we propose an improved Alternative Sinkhorn (AS) algorithm, which effectively combines the advantages of AS algorithm for RD functions \cite{wu2022communication} and the entropy regularized algorithm for Wasserstein Barycenter model \cite{DBLP:conf/icml/CuturiD14}. Numerical experiments demonstrate that the proposed algorithm reaches high precision and efficiency under various circumstances. In many scenarios, it is often the case that the rate is limited, and the goal is to minimize both distortion and perceptual distance under this rate constraint \cite{zhang2021universal}. To address this, computing the Distortion-Rate-Perception (DRP) function becomes a more direct approach to fulfilling this requirement. However, the inclusion of the perceptual term poses challenges for traditional BA type algorithms, similar to those encountered in solving RDP problems. Our proposed method offers a unified framework that enables us to develop an algorithm for DRP functions that is nearly identical to the one used for RDP functions, requiring only minor adjustments.

We apply our numerical methods to an image steganography problem. The growing popularity of generative modeling techniques and digital arts with Non-Fungible Token (NFT), has led to a significant surge in the demand for steganographic communication \cite{volkhonskiy2020steganographic}. This trend is driven by the need to ensure that machine-generated content is distinguishable from other data as well as employing digital watermarking techniques for copyright protection and to safeguard data from counterfeiting. As a prevalent steganography technique, Reverse Data Hiding (RDH) involves embedding secret information into a cover image in a manner that remains imperceptible to the human eye, while enabling lossless recovery of the hidden message and the cover image. With the integration of perceptual measure constraints into the rate-distortion function for RDH problem, our improved AS algorithm facilitates the development of a novel RDH scheme that guarantees perceptual fidelity in addition to the conventional distortion measures. Results show a clear improvement in the perceptual quality, e.g. PSNR and Learned Perceptual Image Patch Similarity (LPIPS) \cite{8578166}, of the stego image over the traditional method under the same embedding rate, which provides assurance of the significance of incorporating perceptual measures in future steganography techniques. Our theoretical findings and numerical techniques offer a compelling basis for the development of steganographic communications that prioritize perceptual quality constraints.

The rest of this paper is organized as follows. Section II introduces the preliminary of RDP problem. In Section III, we establish the RDP function and corresponding WBM-RDP. Section IV investigates the interplay between distortion and perception measures of the RDP function. Section V presents the entropy regularized WBM-RDP and the improved AS algorithm. In Section VI, we provide numerical results which show the advantages of the proposed algorithm. In Section VII, we extend our method to a practical application on steganography. Finally, we conclude the paper in Section VIII.

\section{Preliminary}
\subsection{Notations}
We use capital letters such as $X$ to denote random variables and lower-case letters such as $\bm x$ to denote corresponding instances. The letter $\bm p$ is reserved for distributions. 
For finite alphabet $\mathcal{X},$ let $(\mathcal{X}, \mathcal{F}, \mathcal{P})$ be the probability space and $\mathcal{P}(\mathcal{X})$ denote the probability simplex.
Consider a stationary and memoryless source $\{X_i\}_{i=1}^\infty$ drawn according to distribution $p_{X}(x)\in \mathcal{P}(\mathcal{X})$  and its corresponding reconstruction $\{\hat{X}_i\}_{i=1}^\infty$ that follows $p_{\hat{X}}(\hat{x})\in \mathcal{P}(\hat{\mathcal{X}})$.  Since the alphabets  $\mathcal{X}$ and $\hat{\mathcal{X}}$ are finite, we can denote 
\[p_i := p_X(x_i),\quad r_j := p_{\hat{X}}(\hat{x}_j).\]
The $n$-sequence $(X_1,X_2,\ldots,X_n)$ is denoted by $X^n,$ and $(X_t,X_{t+1},\ldots,X_n)$ is denoted by $X_{t}^n.$ So $p_{X^n}=\prod_{i=1}^n p_{X}$, and here we also assume the reconstruction $\hat{X}$ is also independently identically distributed $p_{\hat{X}^n}=\prod_{i=1}^n p_{\hat{X}}.$  We denote the set of non-negative real number by $\mathbb{R}^{+}$ and the set of positive real numbers by $\mathbb{R}^{++}.$ For simplicity, $\infty$ will hereafter refer to $+\infty$. $\mathds{E}[X]$ denotes the expected value of $X.$
The entropy of a discrete random variable $X$ is defined as $H(\bm X)=-\sum_i p_i\log p_i.$ 
Let $\mathrm{KL}(\cdot, \cdot)$ denote the Kullback–Leibler divergence, so 
$
\mathrm{KL}(p_X, p_{\hat{X}}) = \sum_{i} p_i\log\frac{p_i}{r_i}.
$
The mutual information between two discrete random variables $X$ and $Y$ with joint density $p_{X,Y}$ is 
$I(X,Y)=\mathrm{KL}(p_{X,Y} \Vert p_X\cdot p_Y)=\sum_{\mathcal{X}\times \mathcal{Y}} p_{X,Y}(x,y)\log\frac{p_{X,Y}(x,y)}{p_X(x) p_Y(y)}.$

\subsection{Measures of Distortion and Perceptual Qualities}
\subsubsection{Distortion Measures}
Consider a (measurable) distortion function $\Delta:\mathcal{X}\times \hat{\mathcal{X}}\mapsto [0, D_{\mathrm{max}}].$ We assume that the distortion is bounded by some $D_{\mathrm{max}}>0$ and the following holds.
\begin{assumption}[non-negativity for distortion measures \cite{blau2019rethinking}]\label{assumption-distortion}
    For arbitrary $x\in \mathcal{X}$ and $y\in \hat{\mathcal{X}},$ the single-letter distortion $\Delta(\cdot, \cdot):\mathcal{X}\times \hat{\mathcal{X}}\mapsto [0, D_{\mathrm{max}}]$ satisfies
    \begin{equation*}
    \Delta(x, y) \geq 0 \ \text{ and }\ \Delta(x, y) = 0  \quad\text{ if and only if }\quad x = y.
    \end{equation*}
\end{assumption}
In applications that involve image and video, some commonly used distortion measures include the peak signal-to-noise ratio (PSNR) and the structural similarity index measure (SSIM). In the rate-distortion theory, the single-letter characterization of the relationship between bit rate and distortion is given by the following optimization problem
\begin{align}\label{eq8_0_0}
    R(D)= \min _{p_{\hat{X} \mid X}:\ \mathds{E}[\Delta(X, \hat{X})] \leq D} \quad&I(X, \hat{X}). 
\end{align}

\subsubsection{Perception Measures}
The trade-off between distortion and perceptual quality has been observed empirically in tasks such as super-resolution in computer vision \cite{blau2018perception}, where the evaluation employs both the reconstruction accuracy and the no-reference image quality measures that align better with human-opinion. 
The measurements of perceptual quality often derives from the divergence between the statistics of images, as is suggested in \cite{blau2018perception}, the perceptual measures are defined by comparing the source and reconstructed distributions. Let $p_X$ and $p_{\hat{X}}$ be the input and output distributions defined on the $\sigma$-algebra $(\mathcal{X},\mathcal{F}),$ the perceptual fidelity is measured by a divergence $d(\cdot, \cdot):\mathcal{P}(\mathcal{X})\times \mathcal{P}(\hat{\mathcal{X}})\mapsto [0, P_{\mathrm{max}}],$ such that the following assumptions hold.
\begin{assumption}[non-negativity for perception measures \cite{blau2019rethinking}]\label{assumption-perception}
    For arbitrary distributions $\bm p$ and $\bm r$ defined on the $\sigma$-algebra $(\mathcal{X},\mathcal{F}),$ the perceptual measure satisfies
    \begin{equation*}
    d(\bm{p}, \bm{r}) \geq 0 \ \text{ and }\  d(\bm{p}, \bm{r}) = 0  \ \text{ if and only if }\  \bm{p} \overset{d}{=} \bm{r}.
    \end{equation*}
\end{assumption}
\begin{assumption}[convexity in the second
argument of perception measures\cite{blau2019rethinking}]\label{assumption-convex}
    The perceptual measure is convex in its second argument, i.e., given $\bm{p},\bm{r}_1,\bm{r}_2$ on $(\mathcal{X},\mathcal{F})$ and $\lambda \in [0,1]$, we have
    \begin{equation*}
    d(\bm{p}, \lambda\bm{r}_1+(1-\lambda)\bm{r}_2) \leq \lambda d(\bm{p}, \bm{r}_1) +(1-\lambda)d(\bm{p}, \bm{r}_2) .
    \end{equation*}
\end{assumption}

For now, we focus on the commonly used perceptual measure -- the Wasserstein metric:
\begin{equation} \label{metric_0}
\mathcal{W} (\bm{p},\bm{r}) = \inf _{\bm{\Pi} \in \mathcal{T}\left(\bm{p},\bm{r}\right)} \mathbb{E}_{\bm{\Pi}}[C(X, \hat{X})],
\end{equation} 
where the cost function $C: \mathcal{X} \times \hat{\mathcal{X}} \rightarrow \mathbb{R}_{+}$ is a non-negative function such that $C(x,y)=0\text{ if and only if } x = y$, and the set $\mathcal{T}\left(\bm{p},\bm{r}\right)$ represents the collection of joint probability measures $\bm{\Pi}$ whose marginal measures are $\bm{p}$ and $\bm{r}$. The discrete form can be written as 
\begin{equation} \label{metric}
\begin{aligned}
\mathcal{W} (\bm{p},\bm{r}) =  \min_{\bm{\Pi}}\quad&\sum_{i=1}^M \sum_{j=1}^N \Pi_{i j} c_{i j}  \\
      \text { s.t. }\quad& \sum_{i=1}^M \Pi_{i j}= r_j ,\  \sum_{j=1}^N \Pi_{i j}= p_i,\  \forall i, j,
\end{aligned}
\end{equation}%
 where $(c_{ij})$ denotes the cost matrix between $x_i$ and $\hat{x}_j$. In the case of the Wasserstein-2 metric, $c_{ij} = \|x_i-\hat{x}_j\Vert_2^2$. Other perceptual constraints belong to the class of $f$-divergences \cite{10206646}, including the total variation (TV) distance $\delta(\bm{p}, \bm{r}) = \frac{1}{2}\|\bm{p}-\bm{r} \Vert_1$ and the Kulback-Leibler (KL) divergence $\text{KL}(\bm{p}\| \bm{r}) = \sum_{i=1}^M p_i \left[\log p_{i}-\log r_i\right]$. Generalizations to the TV and KL metrics are rather straightforward and will be discussed in Section V.

The perceptual constraint is intimately related to the theory of coordination \cite{cuff2010coordination}, where two notions of coordination are distinguished, the empirical coordination and the strong coordination. In \cite{chen2022rate} and \cite{niu2023conditional}, these notions are incorporated into the RDP trade-off
, considering two types of corresponding perceptual constraints. The empirical perception compares the empirical distributions of the input and output, while the strong perceptual constraint compares the joint distributions over the block of symbols. In Section~\ref{sec:interplay}, we will consider the information RDP functions which characterize the rate region for these perceptual constraints. 

\subsection{Related Work}

The operational RDP trade-off was first formalized in the work of Blau and Michaeli \cite{blau2019rethinking}. Their approach translated established practices in machine learning into a RDP function, and they showed that the perception-distortion trade-off holds true for any distortion measure. Building on this, Theis and Wagner \cite{theis2021coding} provided a coding theorem for stochastic variable-length codes, addressing both one-shot and asymptotic scenarios. Their analysis considered the encoder and decoder having access to infinite common randomness. Chen et al. \cite{chen2022rate} extended the analysis to the asymptotic regime, providing coding theorems that elucidate the operational meaning of RDP under conditions where the encoder and decoder share common randomness, lack it, or have private randomness. Wagner \cite{wagner2022rate} further advanced the study of RDP trade-offs by formulating a coding theorem that addressed both perfect and near-perfect realism, with "perfect" indicating that the reconstruction distribution matches the source distribution, in the presence of finite common randomness between the encoder and decoder.

The difficulty in deriving analytical expressions of the RDP functions has stimulated the study of computational methods -- closed form solutions to the RDP problem are often intractable, except for some special cases such as the Gaussian source with squared error distortion and Wasserstein-2 metric perception \cite{zhang2021universal}. Recently, a generalized BA type of algorithm \cite{10206646} has been proposed to approximately solve the RDP function with a perception constraint that belongs to the family of $f$-divergences. Several works focus on computing RDP functions using data-driven methods, such as generative adversarial networks, which minimize a weighted combination of the distortion and the perception \cite{blau2019rethinking,zhang2021universal}.

\subsection{Optimal Transport and Wasserstein Barycenter}

The Wasserstein metric as defined in \eqref{metric_0} compares two distributions by determining the optimal transportation plan that minimizes the overall cost associated with moving resources from one distribution to another. With a rich history, it finds applications across various fields, including machine learning \cite{DBLP:journals/siamma/AguehC11,arjovsky2017wasserstein}, incompressible fluid dynamics \cite{baradat2020small}, density function theory \cite{buttazzo2012optimal}, 
and tomographic reconstruction \cite{abraham2017tomographic}. 
A plethora of computational techniques has come to the forefront, encompassing a spectrum of methodologies, including linear programming (LP) \cite{pele2009fast}, combinatorial approaches \cite{santambrogio2015optimal}, proximal splitting methods \cite{combettes2011proximal}, and solutions for Monge-Amphere equations \cite{benamou2014numerical,froese2011convergent}. Moreover, the landscape of optimal transport for high-dimensional distributions has experienced an influx of diverse approximation strategies in recent years. In this landscape, the Sinkhorn algorithm \cite{nutz2022entropic} has garnered significant attention. This algorithm, employed in solving the approximate entropy-regularized problem, plays a pivotal role in computing the Wasserstein metric. 

The Wasserstein Barycenter is the Fréchet mean in the Wasserstein space of probability distributions, serving as an average distribution that minimizes the sum of Wasserstein distances to the given input distributions \cite{DBLP:journals/siamma/AguehC11}. Computing the Wasserstein Barycenter involves solving an optimization problem that minimizes a linear combination of Wasserstein distances, which can be formulated as an LP problem \cite{peyre2019computational}. However, directly applying generic LP solvers becomes prohibitive as the problem size increases. One can resort to first-order methods such as subgradient descent on the dual problem \cite{carlier2015numerical}. Alternatively, one can introduce entropic regularization to smooth the objective and enable gradient-based optimization. As noted by \cite{benamou2015iterative}, the approximative Barycenter problem can be reformulated as an iterative projection that resembles a weighted KL projection. By extending the Sinkhorn method to this reformulated problem, the optimal couplings can be efficiently computed in scaling form \cite{DBLP:conf/icml/CuturiD14}.

Recently, there is a surge in applications of OT methods in  information theory and communication. For example, \cite{bai2023information} introduces and studies an information constrained variation of OT problem which provides an application to  network information theory; \cite{8053918} provides an effective OT-based scheme of flight-time constrained unmanned aerial vehicles as flying base stations that provide wireless service to ground users; Ye et al. \cite{ye2022optimal} explores the intricate relationship between the computation of LM rate and the OT problem. Leveraging this insight, we extend the Sinkhorn algorithm, enhancing computational efficiency significantly when compared to conventional methods. This concept is further expanded to compute the RD function \cite{wu2022communication} and address the information bottleneck problem \cite{chen2023information}, both benefiting from Sinkhorn's computational efficiency. These studies underscore remarkable parallels between information-theoretic problems and OT. For example, the source and target densities in OT are akin to the ones in communication. Additionally, relative entropy, a key concept in information theory, bears similarities to the entropy regularization in OT. These analogies empower us to extend the Sinkhorn method effectively, achieving enhanced computational efficiency in diverse scenarios.

\section{RDP Functions and WBM-RDP}

The information RDP function is defined in \cite{blau2019rethinking} as below.
\begin{definition}[information RDP function \cite{blau2019rethinking}]\label{definition1}
Let $D \in \mathbb{R}^+$ be a given distortion level, $P \in \mathbb{R}^+$ be a given perceptual fidelity level, and $p_X \in \mathcal{P}(\mathcal{X})$ be the source distribution. The information RDP function is defined as follows:
\begin{subequations}\label{eq0}
\begin{align}
R(D, P)= \min _{p_{\hat{X} \mid X}} \quad& I(X, \hat{X}) \label{eq0_a}   \\
 \text { \emph{s.t.}}\quad &\mathds{E}[\Delta(X, \hat{X})] \leq D,\label{eq0_b}   \\  
 &  d\left(p_X, p_{\hat{X}}\right) \leq P, \label{eq0_c}
\end{align}
\end{subequations}%
where the minimization is taken over all conditional distributions. 
\end{definition}

Note that the information RDP function \eqref{eq0} degenerates to RD functions when the constraint \eqref{eq0_c} is removed.
We denote $d_{ij} = \Delta(x_i,\hat{x}_j)$
and $w_{ij}=W(\hat{x}_j \mid x_i)$ for all $1\le i\le M, 1\le j\le N$. Here $W: \mathcal{X} \rightarrow \hat{\mathcal{X}}$ is the transition mapping.
Thus the discrete form of problem (\ref{eq0}) can be written as
\begin{subequations} \label{eq0_0}
\begin{align}
    \min _{\bm{w},\bm{r}} \quad  &\sum_{i=1}^M \sum_{j=1}^N\left(w_{i j} p_i\right)\left[\log w_{i j}-\log r_j\right]   \label{eq0_0_a}\\
    \text { s.t. }\quad  &\sum_{j=1}^N w_{i j}=1,\   \sum_{i=1}^M w_{i j} p_i=r_j,  \  \forall i, j,  \label{eq0_0_b}\\
    &\sum_{i=1}^M \sum_{j=1}^N w_{i j} p_i d_{i j} \leq D,\   \sum_{j=1}^N r_j=1 ,  \label{eq0_0_c}\\
    &d(\bm{p},\bm{r}) \leq P. \label{eq0_0_d}
\end{align}
\end{subequations}

We notice that when the constraint \eqref{eq0_0_d} is in the Wasserstein metric, and under the Kantorovich form we can transform it into a linear optimization problem in a high-dimensional space. Therefore, we start by converting the perception constraint \eqref{eq0_0_d} into a linear constraint in a higher dimensional space. Consequently, all the constraints, \eqref{eq0_0_b}-\eqref{eq0_0_d}, also preserve the simplex structure in the higher-dimensional space. This enables us to derive several results on the interplay between various distortion and perception constraints. The novel formulation also facilitates the design of an efficient algorithm for solving the optimization problem. Next, we formulate an OT problem that is equivalent to the original RDP problem.

\begin{theorem}
\label{thm-0}
The optimal solution to the RDP function, as defined by equation \eqref{eq0_0}, coincides with the optimal solution of the subsequent optimization problem when the distance metric $d(\cdot,\cdot)$ in equation \eqref{eq0_0_d} is specified as the Wasserstein metric. Furthermore, the pair $(\bm{w},\bm{r})$, which represents the optimal solution to the optimization problem delineated in equation \eqref{eq0_1}, is identical to the optimal solution of equation \eqref{eq0_0}.
 \begin{subequations} \label{eq0_1}
\begin{align}
    \min _{\bm{w},\bm{r},\bm{\Pi}} \quad  &\sum_{i=1}^M \sum_{j=1}^N\left(w_{i j} p_i\right)\left[\log w_{i j}-\log r_j\right]  \label{eq0_1_a}\\
    \emph{ \text { s.t. } }\quad &\sum_{j=1}^N w_{i j}=1,\   \sum_{i=1}^M w_{i j} p_i=r_j,  \label{eq0_1_b}\\
    &\sum_{i=1}^M \Pi_{i j}= r_j ,\  \sum_{j=1}^N \Pi_{i j}= p_i, \  \forall i, j, \label{eq0_1_c}\\
    &\sum_{i=1}^M \sum_{j=1}^N w_{i j} p_i d_{i j} \leq D,\   \sum_{j=1}^N r_j=1 , \label{eq0_1_d}\\
    &\sum_{i=1}^M \sum_{j=1}^N \Pi_{i j} c_{i j} \leq P.\label{eq0_1_e}
\end{align}
\end{subequations}
\end{theorem}

We give the proofs in Appendix~\ref{app:proof-thm0}.

Theorem \ref{thm-0} establishes an equivalence between the RDP problem \eqref{eq0_0} and the optimization \eqref{eq0_1}. Also, we observe that the model \eqref{eq0_1} has the Wasserstein Barycenter structure. The optimization variable $r_j$ can be regarded as the Barycenter, and $\text{diag}(\bm{p})\cdot\bm{w}$ and $\bm{\Pi}$ are two couplings between each input $\bm{p}$ and Barycenter $\bm{r}$.

\section{Interplay Between Distortion and Perception}\label{sec:interplay}

With the introduction of the new constraint on the divergence between distributions, a natural question to ask is when this constraint becomes inactive, i.e., the RDP problem can be simplified to the RD problem. We prove the existence of such a critical transition phenomenon by applying the Wasserstein Barycenter framework, deriving several insights into the interplay between the distortion and perception constraints. Furthermore, we show that the corresponding critical transition curves can be effectively computed using our method.

\subsection{Variants of the RDP Functions}
In the classical rate-distortion problem, the RD function exhibits convexity, enabling the formulation of the distortion-rate (DR) function as the inverse of the RD function. In the context of the three-way trade-off involving rate, distortion, and perception, we symmetrically formulate two bivariate functions, namely the perception function with respect to the rate and distortion and the distortion function with respect to the rate and perception, in addition to the rate function.
\begin{definition}
Given rate $R\in \mathbb{R}^+$ and distortion $D\in \mathbb{R}^+$, the information perception-rate-distortion (PRD) function is defined as
\begin{equation}\label{eq8_2}
        \begin{aligned}
        P(R,D) = \min _{p_{\hat{X} \mid X}} \quad&d\left(p_X, p_{\hat{X}}\right)    \\
        \text { \emph{s.t.} }   \quad&  I(X, \hat{X}) \leq R,  \\
        &\mathds{E}[\Delta(X, \hat{X})] \leq D , 
        \end{aligned}
\end{equation}
where the minimization is taken over all conditional distributions.
\end{definition}

\begin{definition}
Given rate $R\in \mathbb{R}^+$ and perceptual fidelity $P\in \mathbb{R}^+$, the information distortion-rate-perception (DRP) function is defined as
\begin{equation}\label{eq8_3}
        \begin{aligned}
        D(R,P) = \min _{p_{\hat{X} \mid X}} \quad& \mathds{E}[\Delta(X, \hat{X})]   \\
        \text { \emph{s.t.} }   \quad&  I(X, \hat{X}) \leq R,  \\
        &d\left(p_X, p_{\hat{X}}\right) \leq P , 
        \end{aligned}
\end{equation}
where the minimization is taken over all conditional distributions.
\end{definition}
We introduce the following shorthand notation for the RDP functions. For $ \alpha \in \mathbb{R}^+ \cup \{\infty\},$ we use $R_\alpha(D)$ to denote the univariate function by restricting the perpetual quality to be less than $\alpha,$ i.e., 
\[
R_\alpha(D):= R(D, P)\vert_{P=\alpha}.
\]
When $\alpha=\infty,$ the perceptual constraint is relaxed, so $R_\infty(D)$ denotes the information rate-distortion function \eqref{eq8_0_0}.
The information distortion-rate function is denoted similarly, where we let
\[
D_\alpha(R):= D(R, P)\vert_{P=\alpha}
\]
for the DRP function $D(R,P)$, and 

\begin{align*}
    D_{\infty}(R)= \min _{p_{\hat{X} \mid X}:\ I(X, \hat{X}) \leq R} \quad&\mathds{E}[\Delta(X, \hat{X})]. 
\end{align*}
As the perceptual constraint is written in the subscript, we reserve the superscript for the distortion constraint, i.e., 
\[
R^\beta(P):= R(D, P)\vert_{D=\beta}.
\]
So
the information rate perception function $R^{\infty}(P)$ and information perception rate function $P^{\infty}(R) $ for a source $p(x)$  can be written as

    \begin{align*}
    R^{\infty}(P)= \min _{p_{\hat{X} \mid X}:\ d\left(p_X, p_{\hat{X}}\right) \leq P}\quad& I(X, \hat{X}), \quad \text{and} 
    \end{align*}

    \begin{align*}
    P^{\infty}(R) = \min _{p_{\hat{X} \mid X}:\ I(X, \hat{X}) \leq R} \quad&d\left(p_X, p_{\hat{X}}\right),  
    \end{align*}
respectively. Note that the rate-perception trade-off is closely related to the problem of communicating probability distributions, where a single-letter characterization of the rate-perception function is given in \cite{kramer2007communicating}. Recently, the information constrained OT problem \cite{bai2023information} also draws attention with applications to the capacity of the relay channel.

\subsection{Degeneration of RDP Functions to RD Functions}

Under Assumption \ref{assumption-distortion}, \ref{assumption-perception} and \ref{assumption-convex}, the following proposition for $R_{\infty}(D)$, $D_{\infty}(R)$ and $R(D,P)$ is given in \cite{blau2019rethinking}.
\begin{proposition}[\cite{blau2019rethinking}]\label{proposition1}
 The following conclusions hold true when Assumption \ref{assumption-distortion}, \ref{assumption-perception} and \ref{assumption-convex} are satisfied:
 \begin{enumerate}
    \item $R_{\infty}(D)$ and $D_{\infty}(R)$ are both convex and non-increasing with respect to corresponding independent variables.
    \item $R(D,P)$ is convex and non-increasing with respect to $D$ and $P$.
\end{enumerate}
\end{proposition}

We precede with the following definition of minimality, which assists our study of the interplay between the constraints and the geometric properties when one of the constraints becomes inactive.
\begin{definition}
A function $f^* \in \{f: \mathcal{D} \rightarrow\mathbb{R}\}$ is said to be \textbf{minimal} in the given function set $\mathcal{F}$ on domain $\mathcal{D} $ if for any function $f \in \mathcal{F}$ we have
\begin{equation*}
    f^*(x) \leq f(x),\quad \forall x \in \mathcal{D}.
\end{equation*}
\end{definition}
Before we delve into our main theorems, we show some basic properties of these variants of the RDP functions in Appendix~\ref{l1}. These lemmas prepare us for the following theorems, which establish the critical transitions on the distortion-perception level.

\begin{theorem}\label{thm-3}
Consider a source distribution $p_X$ defined on the $\sigma$-algebra $(\mathcal{X}, \mathcal{F})$ and the information RDP functions.
The following statements hold true under Assumption \ref{assumption-distortion}, \ref{assumption-perception} and \ref{assumption-convex}: 
\begin{itemize}[noitemsep,topsep=-\parskip]
    \item [1)]  There exists a continuous minimal  function $f: \mathbb{R}^+ \rightarrow \mathbb{R}^+$ satisfying
    \begin{subequations}
      \begin{align}
        \lim_{P \to f(D)}  D(R,P) =& D_{\infty}(R),\quad \forall R \in \mathbb{R}^+,\label{eq10_0_1_a}\\
        \lim_{P \to f(D)}  R(D,P) =& R_{\infty}(D),\quad \forall D \in \mathbb{R}^+.\label{eq10_0_1_b}
    \end{align}  
    \end{subequations}
    
    Furthermore, when $P \geq f(D)$, the perceptual constraints for the distortion function and the rate function become inactive, i.e., 
    \begin{subequations}
        \begin{align}
        D(R,P) =& D_{\infty}(R),\quad \forall R \in \mathbb{R}^+,\\
        R(D,P) =& R_{\infty}(D),\quad \forall D \in \mathbb{R}^+.
    \end{align} 
    \end{subequations}   
    \item [2)]
    When $D \geq D(0,0)$, we have
    \begin{align}
        P(R,D) = P^{\infty}(R)  = 0.
    \end{align}
    \item [3)]
    There exists a convex continuous minimal function $h: \mathbb{R}^+ \rightarrow \mathbb{R}^+$ satisfying
    \begin{align}\label{eq10_0_2}
        \lim_{D \to h(P)}  R(D,P) = R^{\infty}(P),\quad \forall P \in \mathbb{R}^+.
    \end{align}
    Furthermore, when $D \geq h(P)$, we have
    \begin{align}
        R(D,P) = R^{\infty}(P),\quad \forall P \in \mathbb{R}^+, D\geq h(P).
    \end{align}
\end{itemize}
\end{theorem}
We give the proofs in Appendix~\ref{app:proof-thm3}.

\begin{remark}
\label{remark1}
Based on Theorem \ref{thm-3}, we have the inequality $D(R,P) \leq D(0,P) = h(P).$ Consequently, the equality $R(D,P) = R^{\infty}(P)$ is valid exclusively when $D = h(P)$. The condition $D = h(P)$delineates the upper boundary for the distortion-perception trade-off curve (see Fig.~\ref{curve}). 
This curve originates at the point $(D,P) = (h(0),0)$ where
\begin{equation*}
h(0) = \sum_{i=1}^M \sum_{j=1}^N p_i p_j d_{i j},
\end{equation*}
and for $P \geq P(0,D_{\infty}(0))$, we find that $h(P) = D_{\infty}(0).$ This is because $ D_{\infty}(0) \leq D(0,P)$ and the optimal solution $\bm{r}$ for $D_{\infty}(0)$ remains a feasible solution for $D(0,P)$ provided that $P \geq P(0,D_{\infty}(0))$.
\end{remark}
\begin{remark}
\label{remark2}
It is not possible to define a function $g(\cdot)$ such that for $D \geq g(P),$ the relationship $P(R,D) = P^{\infty}(R)$ holds. However, there is a specific region within the distortion-perception trade-off where $P(R,D) = P^{\infty}(R)  = 0$. Identifying the critical value of $D$ corresponds to determining $D(R,0)$ across various values of $R$. This can be written as:
        \begin{subequations} 
        \begin{align*}
    \min _{\bm{w}} \quad  &\sum_{i=1}^M \sum_{j=1}^N w_{i j} p_i d_{i j}  \\
    \text { \emph{s.t.} } \quad &\sum_{j=1}^N w_{i j}=1,\   \sum_{i=1}^M w_{i j} p_i=p_j,  \  \forall i, j,  \\
    &\sum_{i=1}^M \sum_{j=1}^N\left(w_{i j} p_i\right)\left[\log w_{i j}-\log p_j\right]  \leq R. 
        \end{align*}
        \end{subequations}
When $R = 0$, we have $w_{ij}=p_j$ and $D = \sum_{i=1}^M \sum_{j=1}^N p_i p_j d_{i j} = h(0)$. When $R \geq -\sum_{i=1}^M p_i \log(p_i)$, we can always set $w_{ij} = \mathbf{1}_{i = j}$ and $D = 0$. Given the continuity and non-increasing nature of $D(R,0)$, it follows that $D$ is bounded within the interval $D \in [0, h(0)]$. Consequently, the region of interest is characterized by $P = 0$ and $D \in [0, h(0)]$.
\end{remark}
The function $P = f(D)$ represents the transition curve that delineates the boundary between the RDP function and the RD function. This curve initiates at the origin $(0,0)$ and terminates at the point $(D_{\infty}(0),P(0,D_{\infty}(0)))$. 
\begin{remark}
The endpoint of the transition curve, $(D_{\infty}(0),P(0,D_{\infty}(0))),$ is also the intersection of $P = f(D)$ and $D = h(P)$, i.e., $R(D,P) = R_P(D) = R^{\infty}(P) = 0$ at this point.
\end{remark}
\begin{figure}[t]
\centering
\includegraphics[width=\linewidth]{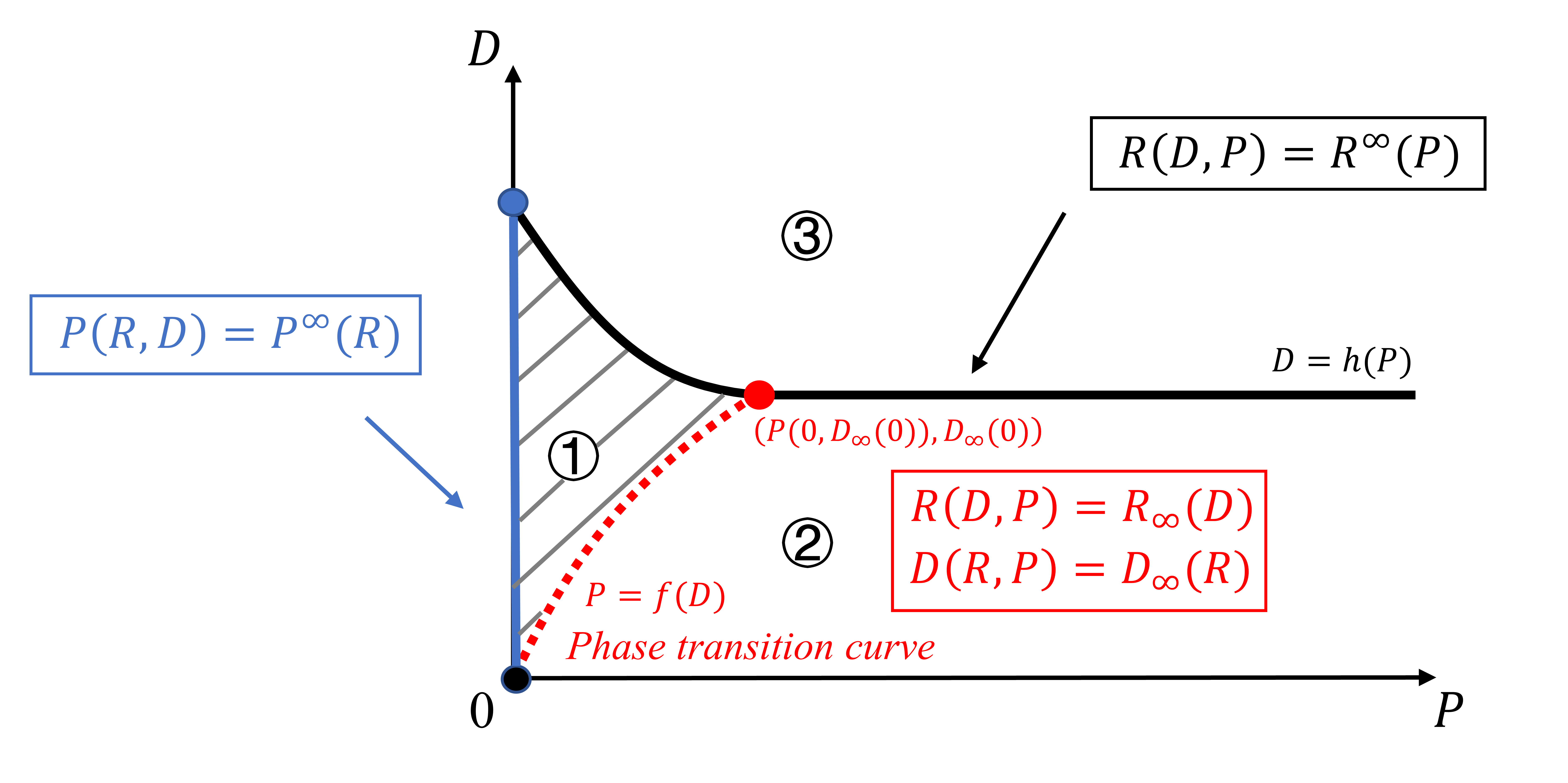}
\caption{The schematic diagram of distortion-perception cross-sections and the corresponding transition curves where one of the constraints is removed.}  \label{curve}
\end{figure}

In Figure \ref{curve}, we present a schematic diagram illustrating the distortion-perception cross-sections. The blue line delineates the region where $P(R,D) = P^{\infty}(R),$ as noted in Remark \ref{remark2}. According to Remark \ref{remark1}, the region marked \normalsize{\textcircled{\scriptsize{3}}}\normalsize\enspace in the figure is achievable for any rate $R\geq 0$, and the black curve represents the boundary defined by $D = h(P)$. In region marked \normalsize{\textcircled{\scriptsize{2}}}\normalsize\enspace , the perceptual constraint is inactive, reducing the RDP problem to an RD problem. Conversely, in the region marked\normalsize{\textcircled{\scriptsize{1}}}\normalsize\enspace both perceptual and distortion constraints influence the rate. The boundary between  \normalsize{\textcircled{\scriptsize{1}}}\normalsize\enspace and \normalsize{\textcircled{\scriptsize{2}}}\normalsize\enspace is the transition curve $P = f(D)$. When the Wasserstein metric is selected as the perceptual measure, the following specific results can be derived regarding the boundary.
\begin{proposition}
\label{thm-4}
If the cost matrix $(c_{ij})$ associated with the perception measure is identical to the distortion matrix $(d_{ij})$, then the endpoint of the transition curve satisfies the conditions $D_{\infty}(0) = P(0,D_{\infty}(0))$ and $ D \geq f(D)$ for $D \in [0,D_{\infty}(0)]$.
\end{proposition}

We give the proofs in Appendix~\ref{app:proof-thm4}. The condition in Proposition~\ref{thm-4} is satisfied by numerous standard distortion and perceptual metrics. For instance, this condition is satisfied when the MSE is used as the distortion measure and the Wasserstein-2 metric is employed as the perceptual measure. Similarly, the condition holds when the Hamming distance serves as the distortion measure and the TV distance is used as the perceptual measure. In both scenarios, the cost matrix $(c_{ij})$ and the distortion matrix $(d_{ij})$ coincide, leading to the satisfaction of the equation $D_{\infty}(0) = P(0,D_{\infty}(0))$ for the respective curves.

\subsection{Extension to the $n$-Letter Case}
We explore how the principles and techniques discussed earlier can be extended to handle operational RDP functions involving sequences of $n$ letters. This extension provides valuable insights for future research within the finite blocklength regime.
In the RDP framework, two scenarios — the one-shot and the asymptotic—have been extensively studied by Theis and Wagner \cite{theis2021coding} and Chen et al. \cite{chen2022rate}, leveraging the strong functional representation lemma \cite{li2018strong} and distributed channel synthesis \cite{DBLP:journals/tit/SaldiLY15a}. It is assumed that various degrees of common randomness between the encoder and decoder, as well as private randomness \cite{hamdi2024the}, are available.

It has been shown that the information RDP function characterizes the rate regions when the perceptual constraints are enforced on the empirical distributions of the output \cite{theis2021coding,chen2022rate,niu2023conditional}. In the preceding subsection, we observe the critical transition that information RDP functions degenerate into RD functions in the $n$-letter case. Additionally, the associated boundary exhibits certain properties. In the following, we also consider perceptual constraints that compare the $n$-product i.i.d. distributions $\prod_{i=1}^n p_X$ and $\prod_{i=1}^n p_{\hat{X}},$ which resembles the strong perceptual constraints. First, it is imperative to make appropriate adjustments to the relevant definitions inspired by the one-shot \cite{liu2024one} result, $k$-shot network code \cite{DBLP:journals/jsait/ZhangGYNB24}, and finite blocklength lossy compression \cite{DBLP:journals/tit/KostinaV12}. However, rather than presenting coding theorems within the finite blocklength regime, our objective is to explore the mathematical consequences of the following operational definition.

\begin{definition}\label{definition7}
Given a distortion $D\in \mathbb{R}^+$ and a perceptual fidelity $P\in \mathbb{R}+$, the $n$-letter information RDP function is 
\begin{equation}\label{eq20}
\begin{aligned}
R^{[n]}(D, P)= \min _{p_{\hat{X} \mid X}} \quad& I(X, \hat{X})    \\
 \text { \emph{s.t.} }\quad &\mathds{E}[\Delta(X, \hat{X})] \leq D,  \\  
 &  d\left(p_{X^n}, p_{\hat{X}^n}\right) \leq P, 
\end{aligned}
\end{equation}%
where $p_{X^n}$ and $p_{\hat{X}^n}$ denote the product measures of the source $p_{X}$ and reconstruction distributions $p_{\hat{X}}$ respectively, i.e., 
\begin{align*}
    p_{X^n}(x_1,\cdots,x_n)&:=\prod_{i=1}^{n}p_X(x_i), \\
    p_{\hat{X}^n}(x_1,\cdots,x_n)&:= \prod_{i=1}^{n}p_{\hat{X}}(x_i).
\end{align*}

\end{definition}
The definitions for $D^{[n]}(R, P), P^{[n]}(R, D), R^{\infty^{[n]}}(P)$ and $P^{\infty^{[n]}}(R)$ follows by simply replacing $d\left(p_X,p_{\hat{X}}\right)$ with $d\left(p_{X^n}, p_{\hat{X}^n}\right)$ respectively. 

In \cite{chen2022rate}, the authors discuss the situation where the perceptual constraint is tensorizable. In this case, the divergence in the perceptual constraint can be decomposed into the sum of the divergences between the marginals, which greatly simplifies the analysis.
\begin{theorem}\label{thm-6}
For perception measure using the Wasserstein distance, the TV distance and the KL divergence, we have
\begin{align*}
        R^{[n]}(D,P) = R(D,\frac{P}{n}).
\end{align*}
\end{theorem}

\begin{proof}
From the Proposition 1 for the given divergences mentioned in \cite{salehkalaibar2024rate}
and the tensorization property of KL divergence \cite{Joyce2011} we can obtain 
\begin{align*}
    \mathcal{W}\left(p_{X^n}, p_{\hat{X}^n} \right) &= n\mathcal{W}\left(p_{X}, p_{\hat{X}} \right), \\
    \delta\left(p_{X^n}, p_{\hat{X}^n} \right) &= n\delta\left(p_{X}, p_{\hat{X}} \right),\\
    \text{KL}\left(p_{X^n}\| p_{\hat{X}^n} \right) &= n\text{KL}\left(p_{X}\| p_{\hat{X}} \right),
\end{align*}
thus $R^{[n]}(D,P)$ in \eqref{eq20} can be written as:
\begin{equation}
\begin{aligned}
R^{[n]}(D, P)= \min _{p_{\hat{X} \mid X}} \quad& I(X, \hat{X})    \\
 \text { s.t. }\quad &\mathds{E}[\Delta(X, \hat{X})] \leq D,   \\  
 &  d\left(p_{X}, p_{\hat{X}}\right) \leq \frac{P}{n}, 
\end{aligned}
\end{equation}
which is also the expression for $R(D, \frac{P}{n})$.
\end{proof}
It is worth noting that as $n\rightarrow \infty$, $R^{[n]}(D,P)$ converges to $R(D,0)$, which signifies perfect perceptual quality \cite{blau2019rethinking}.  This conclusion can be extended to the case where the sequence $\{\hat{X}_i\}$ is \textbf{not} necessarily identically distributed. In such cases, we adjust the corresponding definitions accordingly.
\begin{proposition}
Let $D\in \mathbb{R}^+$ be a given distortion level and $P\in \mathbb{R}+$ be a given perceptual fidelity level. When the sequence $\{\hat{X}_i\}$ is not necessarily identically distributed, we define the $n$-letter information RDP function as
\begin{equation}\label{eq20}
\begin{aligned}
\tilde{R}^{[n]}(D, P)= \min _{p_{\hat{X}_i \mid X}} \quad& \frac{1}{n}\sum_{i = 1}^n I(X, \hat{X}_i)    \\
 \text { \emph{s.t.} }\quad &\frac{1}{n}\sum_{i = 1}^n\mathds{E}[\Delta(X, \hat{X}_i)] \leq D,  \\  
 &  d\left(p_{X^n}, \prod_{i=1}^{n}p_{\hat{X}_i}\right) \leq P. 
\end{aligned}
\end{equation}
Then, for any perception measure using the Wasserstein distance, the TV distance, and the KL divergence, we have:
\begin{align*}
        \tilde{R}^{[n]}(D,P) = R(D,\frac{P}{n}).
\end{align*}
\end{proposition}
\begin{proof}
Inspired by Theorem 7 in \cite{chen2022rate}, for any pair $(X,\hat{X}_i)$ satisfying:
\begin{equation*}
\begin{aligned}
&I(X, \hat{X}_i) = R(D,\frac{P}{n}),     \\
&\mathds{E}[\Delta(X, \hat{X}_i)] \leq D,  \\  
&d\left(p_{X}, p_{\hat{X}_i}\right) \leq \frac{P}{n}.
\end{aligned}
\end{equation*}
We can derive that pair $(X^n,\{\hat{X}_i\}_{i=1}^n)$ satisfy
\begin{equation*}
\begin{aligned}
&\frac{1}{n}\sum_{i = 1}^n I(X, \hat{X}_i) = R(D,\frac{P}{n}),     \\
&\frac{1}{n}\sum_{i = 1}^n\mathds{E}[\Delta(X, \hat{X}_i)] \leq D,  \\  
&d\left(p_{X^n}, \prod_{i=1}^{n}p_{\hat{X}_i}\right) = \sum_{i = 1}^n d\left(p_{X}, p_{\hat{X}_i}\right) \leq P, 
\end{aligned}
\end{equation*}
the last equality is guaranteed since the tensorization property \cite{salehkalaibar2024rate,Joyce2011} only requires independence. Therefore
\begin{equation*}
 \tilde{R}^{[n]}(D,P) \leq R(D,\frac{P}{n}).
\end{equation*}

On the other hand, for any pair $(X^n,\{\hat{X}_i\}_{i=1}^n)$ satisfying
\begin{equation*}
\begin{aligned}
&\frac{1}{n}\sum_{i = 1}^n I(X, \hat{X}_i) = \tilde{R}^{[n]}(D,P),     \\
&\frac{1}{n}\sum_{i = 1}^n\mathds{E}[\Delta(X, \hat{X}_i)] \leq D,  \\  
&d\left(p_{X^n}, \prod_{i=1}^{n}p_{\hat{X}_i}\right) \leq P.
\end{aligned}
\end{equation*}
We denote $T$ a random variable uniformly distributed over $\{1,\cdots , n\}$ and is independent of $X^n$ and $\{\hat{X}_i\}_{i=1}^n$. We can derive that pair $(X_T,\hat{X}_T)$ satisfy
\begin{equation*}
p_{X_T} = p_X, p_{\hat{X}_T} = \frac{1}{n}\sum_{i=1}^n p_{\hat{X}_i},
\end{equation*}
\begin{equation*}
\begin{aligned}
I(X_T, \hat{X}_T) &=  \mathds{E}_T[I(X_T, \hat{X}_T\mid T)]\\
&= \frac{1}{n}\sum_{i = 1}^n I(X, \hat{X}_i) \\
&= \tilde{R}^{[n]}(D,P),
\end{aligned}
\end{equation*}
\begin{equation*}
\begin{aligned}
\mathds{E}[\Delta(X_T, \hat{X}_T)] &=  \mathds{E}_T[\mathds{E}[\Delta(X_T, \hat{X}_T)\mid T]]\\
&= \frac{1}{n}\sum_{i = 1}^n\mathds{E}[\Delta(X, \hat{X}_i)] \leq D,  \\  
\end{aligned}
\end{equation*}
\begin{equation*}
\begin{aligned}
d\left(p_{X_T}, p_{\hat{X}_T}\right) &= d\left(p_X, \frac{1}{n}\sum_{i=1}^n p_{\hat{X}_i}\right)\\
&\leq \frac{1}{n} \sum_{i=1}^n d\left(p_{X_i}, p_{\hat{X}_i}\right) \\
&= \frac{1}{n} d\left(p_{X^n}, \prod_{i=1}^{n}p_{\hat{X}_i}\right) \leq \frac{P}{n}.
\end{aligned}
\end{equation*}
Besides the tensorization property, here we also utilize Assumption \ref{assumption-convex} for convexity. Therefore
\begin{equation*}
 \tilde{R}^{[n]}(D,P) \geq R(D,\frac{P}{n}).
\end{equation*}
This completes the proof.
\end{proof}

In alignment with the insights garnered from the RDP function, we likewise present analogous discoveries concerning some basic properties of RDP functions variants on $n$-product distribution in Appendix \ref{l2}. We can then obtain results similar to Theorem \ref{thm-3}, namely, the existence of the relevant critical transition functions. However, the product measure may not tensorize, so the original convexity may be disrupted.

\begin{proposition}\label{thm-5}
The following statements hold for distortion and perception measures that satisfy Assumption \ref{assumption-distortion}, \ref{assumption-perception} and \ref{assumption-convex}: 
\begin{itemize}[noitemsep,topsep=-\parskip]
    \item [1)]  There exists a minimal function $f^{[n]}: \mathbb{R}^+ \rightarrow \mathbb{R}^+$ such that when $P \geq f^{[n]}(D)$, we have
    \begin{subequations}
        \begin{align}
        D^{[n]}(R,P) =& D_{\infty}(R),\quad \forall R \in \mathbb{R}^+,\label{n1}\\
        R^{[n]}(D,P) =& R_{\infty}(D),\quad \forall D \in \mathbb{R}^+.\label{n2}
    \end{align} 
    \end{subequations}   
    \item [2)]
    When $D \geq D^{[n]}(0,0)$, we have
    \begin{align}
        P^{[n]}(R,D) = P^{\infty^{[n]}}(R) = 0.
    \end{align}
    \item [3)]
    There exists a minimal function $h^{[n]}: \mathbb{R}^+ \rightarrow \mathbb{R}^+$ such that when $D \geq h^{[n]}(P)$, we have
    \begin{align}
        R^{[n]}(D,P) = R^{\infty^{[n]}}(P),\quad \forall P \in \mathbb{R}^+.
    \end{align}
\end{itemize}
\end{proposition}

\section{Entropy Regularized WBM-RDP and Improved AS Algorithm}

We have established the WBM-RDP model \eqref{eq0_1} for computing the RDP functions. However, there are two difficulties in designing a practical algorithm. First, the WBM-RDP \eqref{eq0_1} is not strictly convex on $\bm{\Pi}$. Although the optimal value of \eqref{eq0_1} is unique, the corresponding optimal solutions may vary in the dimensions of $\bm{\Pi}$. Therefore, the convergence and the numerical stability of the AS algorithm designed for RD functions \cite{wu2022communication} cannot be guaranteed. Second, WBM-RDP contains logarithmic terms of the Barycenter objective optimization function $\log r_j$, which is not a standard formulation of the classical Wasserstein Barycenter problems. Moreover, the numerical methods for the Wasserstein Barycenter problem are computationally intensive \cite{DBLP:conf/icml/CuturiD14, peyre2019computational, altschuler2022wasserstein}. 
In this section, we will overcome these difficulties and improve the Alternating Sinkhorn algorithm to solve WBM-RDP efficiently.

\subsection{Entropy Regularized WBM-RDP}
As discussed above, the WBM-RDP \eqref{eq0_1} is not strictly convex on $\bm{\Pi}$, the most direct way \cite{nutz2022entropic} is to introduce an extra entropy regularized term in the objective optimization function \eqref{eq0_1_a}, i.e.,
\begin{equation*}
E(\bm{\Pi})= \sum_{i=1}^M \sum_{j=1}^N \Pi_{i j} \log(\Pi_{i j}).
\end{equation*}
This leads to the entropy regularized WBM-RDP:
\begin{equation}\label{problem}
\begin{aligned}
    \min _{\bm{w},\bm{r},\bm{\Pi}} \quad &\sum_{i=1}^M \sum_{j=1}^N\left(w_{i j} p_i\right)\left[\log w_{i j}-\log r_j\right]  
    + \varepsilon E(\bm{\Pi})    \\
    \text { s.t. }  \quad &\sum_{j=1}^N w_{i j}=1,\   \sum_{i=1}^M w_{i j} p_i=r_j,  \\
    &\sum_{i=1}^M \Pi_{i j}= r_j ,\  \sum_{j=1}^N \Pi_{i j}= p_i,  \forall i, j.  \\
    &\sum_{i=1}^M \sum_{j=1}^N w_{i j} p_i d_{i j} \leq D,\   \sum_{j=1}^N r_j=1 , \\
    &\sum_{i=1}^M \sum_{j=1}^N \Pi_{i j} c_{i j} \leq P.
\end{aligned}
\end{equation}
Here, $\varepsilon>0$ is a newly introduced regularization parameter. By analyzing the Lagrangian of \eqref{problem}, it is not difficult to verify that the model is strictly convex. Fortunately, employing the alternative iteration method to tackle this problem yields closed-form expressions for the dual variables, which accelerates the algorithm while improving the accuracy. Not only that, the following theorem guarantees that the solution to entropy regularized WBM-RDP \eqref{problem} converges to WBM-RDP \eqref{eq0_1} as $\varepsilon \rightarrow 0$.  An important part of the proof is that $E(\bm{\Pi})$ is a strictly convex function, which is crucial for the uniqueness of the optimal solution.
\begin{theorem}
\label{thm-1}
\emph{(Convergence in $\varepsilon$)} The solution $\{ \bm{w}_{\varepsilon},\bm{\Pi}_{\varepsilon},\bm{r}_{\varepsilon} \}$ to (\ref{problem}) converges to the optimal solution with minimal entropy of $E(\bm{\Pi})$ within the set of all optimal solutions to (\ref{eq0_1}), i.e, 
\begin{equation}\label{proof1}
   \{ \bm{w}_{\varepsilon},\bm{\Pi}_{\varepsilon},\bm{r}_{\varepsilon} \} \underset{\varepsilon \rightarrow 0}{\longrightarrow} \operatorname{argmin} \Big \{E(\bm{\Pi}) \Big| \left\{ \bm{w},\bm{\Pi},\bm{r} \right\} \in \mathcal{M}\Big \},
\end{equation}
where $\mathcal{M}$ denotes the set of all optimal solutions to (\ref{eq0_1}).
\end{theorem}

We give the proofs in Appendix~\ref{app:proof-thm1}.

\subsection{The Improved Alternating Sinkhorn Algorithm}

 We construct the Lagrangian function of the regularized WBM-RDP \eqref{problem} by introducing dual variables $\boldsymbol{\alpha},\boldsymbol{\theta} \in \mathbb{R}^M, \boldsymbol{\beta},\boldsymbol{\tau} \in \mathbb{R}^N, \lambda,\gamma \in \mathbb{R}^{+}$ and $\eta \in \mathbb{R}$ \cite{10206982}:
\begin{equation}\label{Lagrangian}
\begin{small}
\begin{aligned}
    &\begin{aligned}
    &\mathcal{L} \left(\bm{w},\bm{\Pi},\bm{r}; \bm{\alpha},\bm{\beta}, \bm{\theta}, \bm{\tau}, \lambda,\eta, \gamma\right)  = \\ &\sum_{i=1}^M{\sum_{j=1}^N} w_{i j} p_i \log \frac{w_{i j}}{r_j}+\varepsilon \sum_{i=1}^M \sum_{j=1}^N \Pi_{i j} \ln \Pi_{i j} +\sum_{i=1}^M \alpha_i\left(\sum_{j=1}^N w_{i j}-1\right)\\ &+\sum_{j=1}^N \beta_j\left(\sum_{i=1}^N w_{i j} p_i-r_j\right)+\sum_{i=1}^M \theta_i\left(\sum_{j=1}^N \Pi_{i j}-p_i\right)\\ &+\sum_{j=1}^N \tau_j\left(\sum_{i=1}^M \Pi_{i j}-r_j\right)+\eta\left(\sum_{j=1}^N r_j-1\right)\\ &+\lambda\left(\sum_{i=1}^M \sum_{j=1}^N w_{i j} p_i d_{i j}-D\right)+\gamma\left(\sum_{i=1}^M \sum_{j=1}^N \Pi_{i j} c_{i j}-P\right).
    \end{aligned}
\end{aligned}
\end{small}
\end{equation}

Here we note that the Lagrangian function \eqref{Lagrangian} is convex with respect to each variable. 
Furthermore, we can improve the algorithm by designing the directions of the alternating iterations according to how the variables appear in \eqref{Lagrangian}.
Next, we give the detailed derivation of our algorithm.

\subsubsection{Updating $w$ and Dual Variables}
Taking the derivative of $\mathcal{L} \left(\bm{w},\bm{\Pi},\bm{r} ; \bm{\alpha},\bm{\beta}, \bm{\theta}, \bm{\tau}, \lambda,\eta, \gamma\right)$ with respect to the primal variable $\bm{w}$, one obtains the representation of $\bm{w}$  by dual variables
\begin{equation} \label{eq1}
    w_{i j}=\exp \left(-\alpha_i / p_i-\beta_j-\lambda d_{i j}-1+\log r_j\right).
\end{equation}

Denote $\phi_i=\exp (-\alpha_i / p_i-1/ 2), \psi_j = \exp (-\beta_j-1/2)$ and  $K_{i j}=\exp (-\lambda d_{i j})$, (\ref{eq1}) yields
\begin{equation*}
    w_{i j}=\phi_iK_{i j}\psi_jr_j.
\end{equation*}

Substituting the above formula into the corresponding boundary condition, we obtain
\begin{gather*}
\phi_i \sum_{j=1}^N K_{i j} \psi_j r_j=1, \quad i=1, \cdots, M, \\
\psi_j  \sum_{i=1}^M K_{i j} \phi_i p_i=1, \quad j=1, \cdots, N .
\end{gather*}

Thus, we can alternatively update $\psi_j, \phi_i$:
\begin{equation*}
    \psi_j \leftarrow 1 / \sum_{i=1}^M K_{i j} \phi_i p_i, \quad \phi_i \leftarrow 1 / \sum_{j=1}^N K_{i j} \psi_j r_j.
\end{equation*}
After updating $\bm{\phi}, \bm{\psi}$, we also finish updating the associated dual variables $\boldsymbol{\alpha,\beta}$.

Taking the derivative of $\mathcal{L} \left(\bm{w},\bm{\Pi},\bm{r}; \bm{\alpha},\bm{\beta}, \bm{\theta}, \bm{\tau}, \lambda,\eta, \gamma\right)$  with respect to $ \lambda$, we have the following condition for $\lambda \in \mathbb{R}^{+}$:
\begin{equation}
    F(\lambda) \triangleq \sum_{i=1}^M\sum_{j=1}^N d_{i j}p_i \phi_i e^{-\lambda d_{i j}} \psi_j r_j - D = 0.
\end{equation}

Note that the derivative of $F(\lambda)$ is always negative:
\begin{equation*}
    F'(\lambda)=-\sum_{i=1}^M\sum_{j=1}^N d_{i j}^2p_i \phi_i e^{-\lambda d_{i j}} \psi_j r_j < 0.
\end{equation*}

The monotonicity of the above function $F(\lambda)$ ensures that we can efficiently determine the roots of $F(\lambda) = 0$ using Newton's method.
\begin{remark} \label{feas}
    Similar to \cite{ye2022optimal}, we need to discuss the feasibility of the iteration of updating the dual variables $\lambda$. Depending on the value of $F(0)$, there are two cases:
    \begin{itemize}
        \item $F(0) > 0:$ In this case, $F(\lambda) = 0$  has a unique solution on $\mathbb{R}^{+}$ since $F'(\lambda) \leq 0$. Thus the constraint of distortion is obviously satisfied.
        \item $F(0) = 0:$ In this case, the constraint of distortion is already satisfied. We only need to set $\lambda = 0$ instead of solving $F(\lambda) = 0$.
    \end{itemize}
\end{remark}

\subsubsection{Updating $\Pi$ and Dual Variables}

Similar to 1), taking the derivative of $\mathcal{L} \left(\bm{w},\bm{\Pi},\bm{r} ; \bm{\alpha},\bm{\beta}, \bm{\theta}, \bm{\tau}, \lambda,\eta, \gamma\right)$ with respect to the primal variable $\bm{\Pi}$, one obtains the representation of $\bm{\Pi}$  by dual variables
\begin{equation}\label{eq2}
    \Pi_{i j}=\exp \left(-\theta_i / \varepsilon -\tau_j/ \varepsilon-\gamma c_{i j}/\varepsilon -1\right).
\end{equation}

Denote $\xi _i=\exp(-\theta_i / \varepsilon-1/2), \varphi _j = \exp(-\tau_j/ \varepsilon-1/2)$ and $M_{i j}=\exp(-\gamma  c_{i j}/ \varepsilon)$, (\ref{eq2}) yields
\begin{equation*}
        \Pi_{i j}=\xi _iM_{i j}\varphi_j.
\end{equation*}

Substituting the above formula into the corresponding boundary condition, we get
\begin{gather*}
\xi_i \sum_{j=1}^N M_{i j}  \varphi_j=p_i, \quad i=1, \cdots, M, \\
\varphi_j  \sum_{i=1}^M M_{i j} \xi_i=r_j, \quad j=1, \cdots, N .
\end{gather*}

Thus, we can alternatively update $\varphi_j, \xi_i$:
\begin{equation*}
    \varphi_j \leftarrow r_j / \sum_{i=1}^M M_{i j} \xi_i, \quad \xi_i \leftarrow p_i / \sum_{j=1}^N M_{i j} \varphi_j.
\end{equation*}
After updating $\bm{\xi}, \bm{\varphi}$, we also finish updating the associated dual variables $\boldsymbol{\theta,\tau}$. 

Taking the derivative of $\mathcal{L} \left(\bm{w},\bm{\Pi},\bm{r}; \bm{\alpha},\bm{\beta}, \bm{\theta}, \bm{\tau}, \lambda,\eta, \gamma\right)$  with respect to $ \gamma$, we have the following condition for $\gamma \in \mathbb{R}^{+}$:
\begin{equation}
    G(\gamma) \triangleq  \sum_{i=1}^M\sum_{j=1}^N c_{i j} \xi_i e^{-\gamma  c_{i j}/ \varepsilon} \varphi_j - P = 0 .
\end{equation}

We can utilize Newton's method to solve the above equation for the root $\gamma$, given that $G(\gamma)$ is monotonically decreasing, i.e.,
\begin{equation*}
    G'(\gamma)=-\sum_{i=1}^M\sum_{j=1}^N c_{i j}^2\xi_i e^{-\gamma  c_{i j}/ \varepsilon} \varphi_j/ \varepsilon < 0.
\end{equation*}
Similar to Remark \ref{feas}, we can also give the feasibility here, and here we omit those discussions. 

\subsubsection{Updating $r$ and Dual Variable}

We take the derivative of $\mathcal{L} \left(\bm{w},\bm{\Pi},\bm{r} ; \bm{\alpha},\bm{\beta}, \bm{\theta}, \bm{\tau}, \lambda,\eta, \gamma\right)$ with respect to the primal variable $\bm{r}$, which yields the representation of $\bm{r}$  by dual variables $\bm{\beta}$,$\bm{\tau}$,$\eta$:
\begin{equation}
   r_j =(\sum_{i=1}^M w_{i j} p_i) / (\eta-\beta_j-\tau_j).
\end{equation}

Substituting the above formula into the corresponding boundary condition, we get
\begin{equation}\label{eq3}
    H(\eta) \triangleq \sum_{j=1}^N\left[\left(\sum_{i=1}^M w_{i j} p_i\right) /\left(\eta-\beta_j-\tau_j\right)\right]-1=0.
\end{equation}

Note that the derivative of $H(\eta)$ is also always negative:
\begin{equation*}
    H'(\eta)-\sum_{j=1}^N\left[\left(\sum_{i=1}^M w_{i j} p_i\right) /\left(\eta-\beta_j-\tau_j\right)^2\right] < 0,
\end{equation*}
thus function  $H(\eta)$  is monotonic in each monotone interval.
\begin{remark} \label{feas2}
    According to (\ref{eq3}) and the natural condition for marginal distribution $r_j \geq 0,  \forall j$, we must have $\eta >\max_j(\beta_j+\tau_j).$ For sufficiently small $\epsilon$, it is clear that $H(\eta)>0$ for $\eta = \max_j(\beta_j+\tau_j)+\epsilon$. When $\eta \rightarrow \infty, H(\eta) \rightarrow -1$, then the function $H(\eta)$ has and only has one root in $(\max_j(\beta_j+\tau_j),+\infty)$. Therefore, we can also use Newton's method to efficiently find the root of $H(\eta)$in the interval $(\max_j(\beta_j+\tau_j),+\infty)$ with only a few iterations.
\end{remark}

Thereout, we get the numerical algorithm for WBM-RDP. Our proposed algorithm adopts a similar alternating iteration technique to the Alternating Sinkhorn algorithm for RD functions. 
However, a crucial distinction lies in the operation we perform, which involves altering the projection direction of different joint distribution variables, i.e., $\bm{w}$ and $\bm{\Pi}$, which is essentially different from the AS algorithm for RD functions. 
Thus we name it the Improved Alternating Sinkhorn Algorithm  and the pseudo-code is presented in Algorithm \ref{alg:OT_rdp}.

\begin{algorithm}[t]
 \caption{The Improved Alternating Sinkhorn Algorithm}
 \label{alg:OT_rdp}
 \begin{algorithmic}[1]
  \REQUIRE Distortion measure $d_{ij}$, marginal distribution $p_{i}$, cost matrix $c_{ij}$, regularization parameter $\varepsilon$, \\maximum iteration number $max\_iter$.
  \ENSURE Minimal value  $\sum_{i=1}^{M} \sum_{j=1}^{N} (w_{i j}p_{i}) \left[\log w_{i j}-\log r_{j}\right]+ \varepsilon E(\bm{\Pi})$  with respect to variables $\bm{w}$, $\bm{r}$ and $\bm{\Pi}$.
  \STATE \textbf{Initialization:} $\bm{\phi}, \bm{\xi} = \mathbf{1}_{M}, \bm{\psi},\bm{\varphi} = \mathbf{1}_{N}, \lambda,\gamma=1;$
  \STATE Set $K_{ij} \gets \exp(-\lambda d_{ij})$
  \STATE Set $M_{ij} \gets \exp(-\gamma c_{ij}/\varepsilon)$
  \FOR{$\ell = 1 : max\_iter$}
  \STATE $\psi_{j} \gets 1/\sum_{i=1}^{M}K_{ij}\phi_{i}p_{i}, \quad j=1,\cdots,N$
  \STATE $\phi_{i} \gets 1/\sum_{j=1}^{N}K_{ij}\psi_{j}r_{j}, \quad i=1,\cdots,M$
  \STATE Solve $F(\lambda) = 0$ for $\lambda\in\mathbb{R}^{+}$ with Newton's method
  \STATE Update $K_{ij} \gets \exp(-\lambda d_{ij})$ and $w_{ij}\gets\phi_{i}K_{ij}\psi_{j}r_{j}$
  \STATE $\varphi_j \gets r_j / \sum_{i=1}^M M_{i j} \xi_i, \quad j=1,\cdots,N$
  \STATE $\xi_i \gets p_i / \sum_{j=1}^N M_{i j} \varphi_j, \quad i=1,\cdots,M$
  \STATE Solve $G(\gamma) = 0$ for $\gamma\in\mathbb{R}^{+}$ with Newton's method
  \STATE Update $M_{ij} \gets \exp(-\gamma c_{ij}/\varepsilon)$ 
  \STATE Solve $H(\eta) = 0$ for $\eta\in\mathbb{R}$ with Newton's method
  \STATE Update $r_{j}\gets\left(\sum_{i=1}^{M} w_{i j}p_{i} \right)\Big/\left(\eta-\beta_{j}-\tau_{j}\right)$ 
  \ENDFOR
  \RETURN $\sum_{i=1}^{M}\sum_{j=1}^{N} \left(\phi_{i}p_{i}K_{ij}\psi_{j}r_{j}\right)\left[\log \left(\phi_{i}K_{ij}\psi_{j}\right)\right]$
 \end{algorithmic}
\end{algorithm}
\begin{algorithm}[t]
 \caption{The Improved Alternating Sinkhorn Algorithm for DPR functions}
 \label{alg:OT_drp}
 \begin{algorithmic}[1]
  \REQUIRE Distortion measure $d_{ij}$, marginal distribution $p_{i}$, cost matrix $c_{ij}$, regularization parameter $\varepsilon$, \\maximum iteration number $max\_iter$.
  \ENSURE Minimal value  $\sum_{i=1}^M \sum_{j=1}^N w_{i j} p_i d_{i j} +\varepsilon E(\bm{\Pi})$  with respect to variables $\bm{w}$, $\bm{r}$ and $\bm{\Pi}$.
  \STATE \textbf{Initialization:} $\bm{\phi}, \bm{\xi} = \mathbf{1}_{M}, \bm{\psi},\bm{\varphi} = \mathbf{1}_{N}, \lambda,\gamma=1;$
  \STATE Set $K_{ij} \gets \exp(-\lambda d_{ij}/\gamma)$
  \STATE Set $M_{ij} \gets \exp(-\gamma c_{ij}/\varepsilon)$
  \FOR{$\ell = 1 : max\_iter$}
  \STATE $\psi_{j} \gets 1/\sum_{i=1}^{M}K_{ij}\phi_{i}p_{i}, \quad j=1,\cdots,N$
  \STATE $\phi_{i} \gets 1/\sum_{j=1}^{N}K_{ij}\psi_{j}r_{j}, \quad i=1,\cdots,M$
  \STATE Solve $G_2(\gamma) = 0$ for $\lambda\in\mathbb{R}^{+}$ with Newton's method
  \STATE Update $K_{ij} \gets \exp(-\lambda d_{ij}/\gamma)$ and $w_{ij}\gets\phi_{i}K_{ij}\psi_{j}r_{j}$
  \STATE $\varphi_j \gets r_j / \sum_{i=1}^M M_{i j} \xi_i, \quad j=1,\cdots,N$
  \STATE $\xi_i \gets p_i / \sum_{j=1}^N M_{i j} \varphi_j, \quad i=1,\cdots,M$
  \STATE Solve $F_2(\lambda) = 0$ for $\lambda\in\mathbb{R}^{+}$ with Newton's method
  \STATE Update $M_{ij} \gets \exp(-\lambda c_{ij}/\varepsilon)$ 
  \STATE Solve $H_2(\eta) = 0$ for $\eta\in\mathbb{R}$ with Newton's method
  \STATE Update $r_{j}\gets\gamma\left(\sum_{i=1}^{M} w_{i j}p_{i} \right)\Big/\left(\eta-\beta_{j}-\tau_{j}\right)$ 
  \ENDFOR
  \STATE \textbf{end}
  \RETURN $\sum_{i=1}^M \sum_{j=1}^N w_{i j} p_i d_{i j} $
 \end{algorithmic}
\end{algorithm}

\begin{remark}
    If the perception measure in \eqref{eq0_0_d} is substituted by TV distance, we only need to set the cost matrix in \eqref{eq0_1_e} as $c_{ij} = \mathbf{1}_{i \neq j}$ (see Eq. (6.11) of \cite{villani2009optimal}). Then our improved AS algorithm is still applicable.   
\end{remark}

\begin{remark}
   If the perception measure in \eqref{eq0_0_d} is substituted by KL divergence, our improved AS algorithm would be simpler. The Sinkhorn iteration in step 2) can be omitted since $\bm{\Pi}$ does not need to be introduced. In step 2) the $G(\gamma)$ is substituted by
    \begin{equation}
       G_{\text{KL}}(\gamma) \triangleq \sum_{j=1}^M p_j \log \left(\Big(\gamma p_j + \sum_{i=1}^M w_{i j} p_i \Big) \Big/(\eta -\beta_j)\right) - T,
    \end{equation}
    where $T = \sum_{i=1}^M p_i \log p_{i} - P$ and $H(\eta)$ is substituted by
    \begin{equation}
        H_{\text{KL}}(\eta) \triangleq \sum_{j=1}^M \left(\gamma p_j + \sum_{i=1}^M w_{i j} p_i \right) \Big/(\eta -\beta_j)-1,
    \end{equation}
     and $r_j = (\gamma p_j + \sum_{i=1}^M w_{i j} p_i ) /(\eta -\beta_j) $. 
   
\end{remark}

\subsection{The Improved AS Algorithm for DRP Functions}

In practical applications, the more common scenario is that the rate is limited and we need to minimize distortion and perception distance under this rate constraint \cite{zhang2021universal}. To address this, computing the DRP function becomes a more direct approach to fulfilling this requirement. The DRP function is akin to the RDP function, with the key difference being the swapping of the rate objective and the distortion constraint, which results in a significant increase in the complexity of solving the problem due to the introduction of new nonlinear constraints from the rate expression. Additionally, the inclusion of the perceptual term poses challenges for traditional BA type algorithms, similar to those encountered in solving RDP problems. Fortunately, our proposed method offers a unified framework that enables us to develop an algorithm for DRP functions that is nearly identical to the one used for RDP functions, requiring only minor adjustments. Similar to the Wasserstein Barycenter Model for RDP, we can also apply the Wasserstein Barycenter model to DRP functions, as shown in equation \eqref{eq8_3}.
\begin{equation} \label{eq6}
\begin{aligned}
    \min _{\bm{w},\bm{r},\bm{\Pi}} \quad  &\sum_{i=1}^M \sum_{j=1}^N w_{i j} p_i d_{i j}  \\
    { \text { s.t. } } \quad&\sum_{j=1}^N w_{i j}=1,\   \sum_{i=1}^M w_{i j} p_i=r_j,  \\
    &\sum_{i=1}^M \Pi_{i j}= r_j ,\  \sum_{j=1}^N \Pi_{i j}= p_i, \  \forall i, j, \\
    &\sum_{i=1}^M \sum_{j=1}^N \Pi_{i j} c_{i j} \leq P,\   \sum_{j=1}^N r_j=1 , \\
    & \sum_{i=1}^M \sum_{j=1}^N\left(w_{i j} p_i\right)\left[\log w_{i j}-\log r_j\right] \leq R.
\end{aligned}
\end{equation}

The numerical solution of equation \eqref{eq6} also requires the entropy regularized term $E(\bm{\Pi})$. From this, we directly give the corresponding Lagrangian function by introducing dual variables $\boldsymbol{\alpha},\boldsymbol{\theta} \in \mathbb{R}^M, \boldsymbol{\beta},\boldsymbol{\tau} \in \mathbb{R}^N, \lambda,\gamma \in \mathbb{R}^{+}$ and $\eta \in \mathbb{R}$:
\begin{equation}\label{eq7}
\begin{small}
\begin{aligned}
    &\begin{aligned}
    &\mathcal{L}_2 \left(\bm{w},\bm{\Pi},\bm{r}; \bm{\alpha},\bm{\beta}, \bm{\theta}, \bm{\tau}, \lambda,\eta, \gamma\right) =\\ &\sum_{i=1}^M{\sum_{j=1}^N}  w_{i j} p_i d_{i j}+\varepsilon \sum_{i=1}^M \sum_{j=1}^N \Pi_{i j} \ln \Pi_{i j}  +\sum_{i=1}^M \alpha_i\left(\sum_{j=1}^N w_{i j}-1\right)\\ &+\sum_{j=1}^N \beta_j\left(\sum_{i=1}^N w_{i j} p_i-r_j\right)+\sum_{i=1}^M \theta_i\left(\sum_{j=1}^N \Pi_{i j}-p_i\right)\\ &+\sum_{j=1}^N \tau_j\left(\sum_{i=1}^M \Pi_{i j}-r_j\right)+\eta\left(\sum_{j=1}^N r_j-1\right)\\
    &+\lambda\left(\sum_{i=1}^M \sum_{j=1}^N \Pi_{i j} c_{i j}-P\right)+\gamma\left( \sum_{i=1}^M \sum_{j=1}^N w_{i j} p_i \log \frac{w_{i j}}{r_j}-R\right).
    \end{aligned}
\end{aligned}
\end{small}
\end{equation}

Thus problem (\ref{eq7}) can be efficiently solved by our improved Alternating Sinkhorn algorithm. Here we  sketch the main ingredients and the pseudo-code is presented in Algorithm \ref{alg:OT_drp}.

\begin{enumerate}
  \item [1)] 
   Update $\bm{w}$ and associated dual variables $\boldsymbol{\alpha,\beta},\lambda,\gamma$ while fixing  $\bm{r}, \bm{\Pi}$ as constant parameters. We can use Sinkhorn algorithm to alternatively update $\alpha_i, \beta_j$:
  \begin{equation*}
    \psi_j \leftarrow 1 / \sum_{i=1}^M K_{i j} \phi_i p_i, \quad \phi_i \leftarrow 1 / \sum_{j=1}^N K_{i j} \psi_j r_j,
\end{equation*}
where $\phi_i=\exp (-\frac{\alpha_i}{p_i\gamma}-\frac{1}{2}), \psi_j = \exp (-\frac{\beta_j}{\gamma}-\frac{1}{2})$ and  $K_{i j}=\exp (-\frac{ d_{i j}}{\gamma})$. We can use the Newton’s method to find the root of the following monotonic single-variable function on $\mathbb{R}^{+}$:
\begin{equation}
    G_2(\gamma) \triangleq  \sum_{i=1}^M\sum_{j=1}^N p_iw_{ij}\log (\phi_i K_{ij}\psi_j) - R = 0,
\end{equation}
  \item [2)]
   Update $\bm{\Pi}$ and associated dual variables $\boldsymbol{\theta,\tau}$ while fix  $\bm{r}, \bm{w}$ as constant parameters. Using the similar technique in 1), we first use Sinkhorn algorithm to alternatively update $\theta_i, \tau_j$:
\begin{equation*}
    \varphi_j \leftarrow r_j / \sum_{i=1}^M M_{i j} \xi_i, \quad \xi_i \leftarrow p_i / \sum_{j=1}^N M_{i j} \varphi_j.
\end{equation*}
where $\xi _i=\exp(-\theta_i / \varepsilon-1/2), \varphi _j = \exp(-\tau_j/ \varepsilon-1/2)$ and $M_{i j}=\exp(-\lambda c_{i j}/ \varepsilon)$.
We can also use the Newton’s method to find the root of the following monotonic single-variable function on $\mathbb{R}^{+}$:
\begin{equation}
    F_2(\lambda) \triangleq  \sum_{i=1}^M\sum_{j=1}^N c_{i j} \xi_i e^{-\lambda  c_{i j}/ \varepsilon} \varphi_j - P = 0 .
\end{equation}
  \item [3)]
   Update $\bm{r}$ and associated dual variables $\eta$ while fix $\bm{w}, \bm{\Pi}$ as constant parameters. We can use Newton's method to find the root of the following single-variable function on its largest monotone interval $(\max_j(\beta_j+\tau_j),+\infty)$:
   \begin{equation}
    H_2(\eta) \triangleq \gamma\sum_{j=1}^N\left[\left(\sum_{i=1}^M w_{i j} p_i\right) /\left(\eta-\beta_j-\tau_j\right)\right]-1=0,
\end{equation}
and we finally get $ r_j =\gamma(\sum_{i=1}^M w_{i j} p_i) / (\eta-\beta_j-\tau_j)$.
\end{enumerate}

\begin{remark}
    The feasibility of the iteration of updating the dual variables $\gamma,\lambda$ and $\eta$ is similar to Remark \ref{feas} and \ref{feas2}, here we omit those discussions.
\end{remark}

\begin{remark}
    Here we do not delve into the numerical solution of the PRD functions \eqref{eq8_2}. 
    Given our ability to solve RDP and DRP functions, we have already provided a comprehensive characterization of the trade off among these three indexes. 
\end{remark}

\section{Numerical Experiment}

In this section, we numerically study the validity of the WBM-RDP and the improved AS algorithm. We compute RDP functions under two settings with different perception measures: one is the binary source with Hamming distortion \cite{blau2019rethinking}, and the other is the Gaussian source with squared error distortion \cite{zhang2021universal}. Moreover, the above two settings have analytical expressions \cite{blau2019rethinking,zhang2021universal} when the perceptual constraints are TV distance and Wasserstein-2 metric, respectively. We use them to test the accuracy of our method.

\begin{itemize}
    \setlength{\itemsep}{12pt}
    \item The binary source $X \sim \mathrm{Be}(p)$ with Hamming distortion and TV distance perception: WLOG, we assume $p \leq 1/2$,
    \begin{itemize}
        \item when $P>p$:
        \begin{small}
        \begin{equation*}
    R(D,P) = R(D, \infty)= \begin{cases}H_b(p)-H_b(D), & D \in[0, p) \\ 0, & D \in[p, \infty)\end{cases},
        \end{equation*}
        \end{small}
    \item when $P\leq p$:
    \begin{small}
        \begin{equation*}
        \begin{aligned}
        R(D, P)= \begin{cases}H_b(p)-H_b(D), & D \in \mathcal{S}_1 \\
        2 H_b(p)+H_b(p-P) \\-H_t\left(\frac{D-P}{2}, p\right)-H_t\left(\frac{D+P}{2}, q\right), & D \in \mathcal{S}_2 \\
        0, & D \in \mathcal{S}_3\end{cases},
        \end{aligned}
        \end{equation*}
    \end{small}
    \end{itemize}
    where $q = 1-p$. $H_b(\alpha)$ denotes the entropy of a binary random variable with probabilities $\alpha,1-\alpha$, $H_t(\alpha,\beta)$ denotes the entropy of a ternary random variable with probabilities $\alpha,\beta,1-\alpha-\beta$, and 
    \begin{equation*}
    \begin{aligned}
    &\mathcal{S}_1 =  [0,\frac{P}{1-2(p-P)}  ),\\ &\mathcal{S}_2 =  [\frac{P}{1-2(p-P)},2pq-(q-p)P  ),\\  &\mathcal{S}_3 =  [2pq-(q-p)P,+\infty  ).
    \end{aligned}
    \end{equation*}
    
    \item The Gaussian source $X\sim \mathcal{N} (\mu_X, \sigma^2_X )$ with squared error distortion and Wasserstein-2 metric perception:
    \begin{itemize}
    \item when $\sqrt{P} \leq \sigma_X-\sqrt{\left|\sigma_X^2-D\right|}$:
    \begin{small}
    \begin{equation*}
    \!\!\!\!\!\! R(D, P)= \frac{1}{2} \log \frac{\sigma_X^2\left(\sigma_X-\sqrt{P}\right)^2}{\sigma_X^2\left(\sigma_X-\sqrt{P}\right)^2-\left(\frac{\sigma_X^2+\left(\sigma_X-\sqrt{P}\right)^2-D}{2}\right)^2},
    \end{equation*}
    \end{small}
    \item when $\sqrt{P} > \sigma_X-\sqrt{\left|\sigma_X^2-D\right|}$:
    \begin{equation*}
    R(D, P)=  \max \left\{\frac{1}{2} \log \frac{\sigma_X^2}{D}, 0\right\}.
    \end{equation*}
    \end{itemize}
\end{itemize}

For the binary source, we can directly compute the result since we can set the discrete distribution $p$ beforehand. As for Gaussian source, we first truncate the sources into an interval $[-S, S]$ and then discretize it by a set of uniform grid points $\{ x_i\}_{i=1}^{N}$ whose adjacent spacing is $\delta = 2S/({N-1})$, i.e.,
\begin{equation*}
    x_i = -S + (i-1)\delta, \quad i = 1,\cdots, N.
\end{equation*}
The corresponding distribution $\bm{p}$ of the Gaussian source can then be denoted by
\begin{equation*}
p_i = F(x_i+\frac{\delta}{2})-F(x_i-\frac{\delta}{2}), \quad i = 1,\cdots, N,
\end{equation*}
where $F(x)$ denotes the distribution of the Gaussian source. Unless otherwise specified, we take $p = 0.1$ for the binary source and $S = 8, \delta = 0.5, \mu = 0, \sigma = 2$ for the Gaussian source. For our entropy regulation method, we set $\varepsilon = 0.01$.

\begin{figure}[t]
\centering
\includegraphics[width=\linewidth]{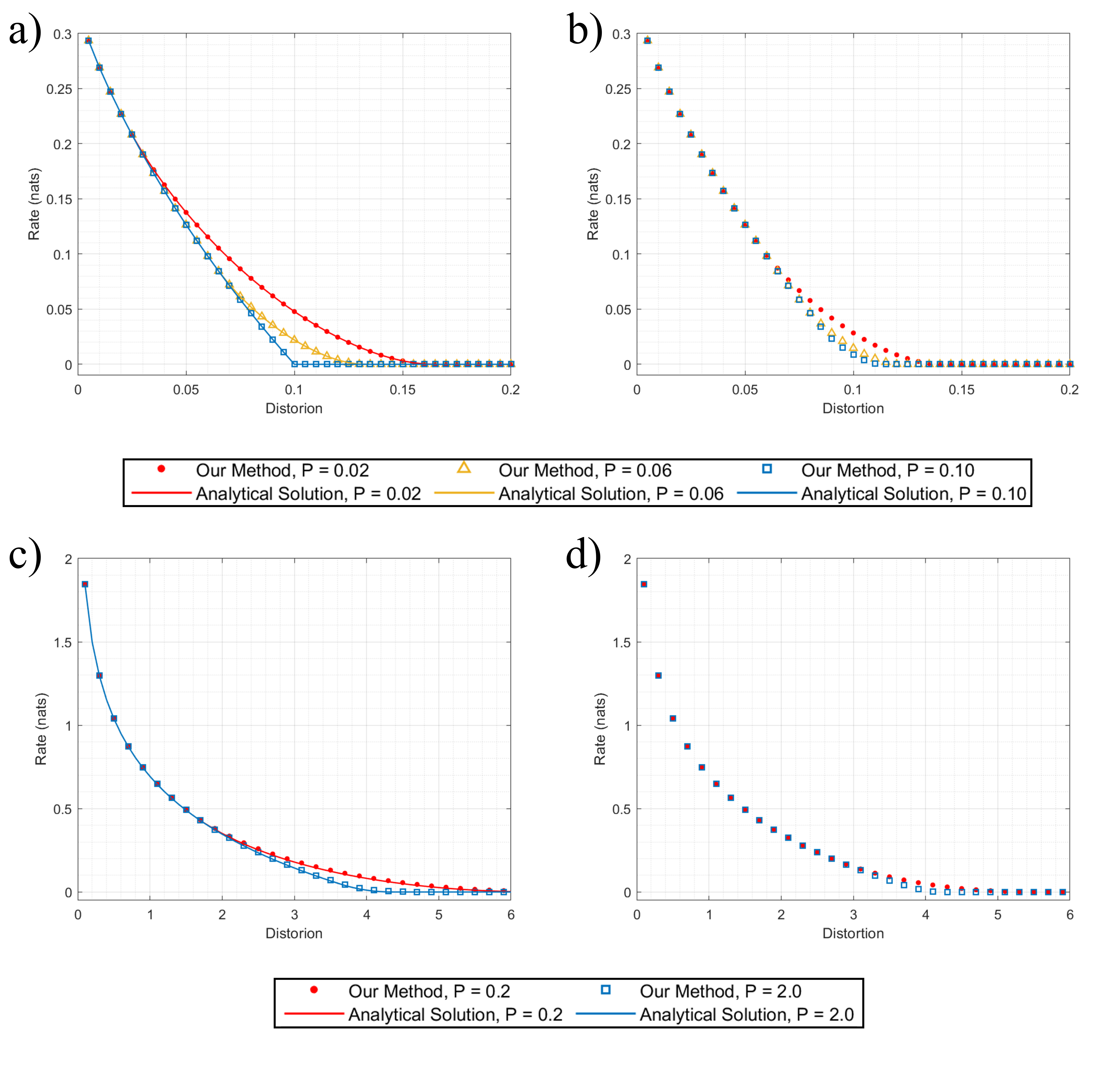}
\caption{The rate-distortion-perception functions obtained by our method and a) the binary source with TV distance and its analytical solution; b) the binary source with KL divergence; c) the Gaussian source with Wasserstein-2 metric and its analytical solution;  d) the Gaussian source with TV distance.}  \label{fig1}
\end{figure}

In Fig. \ref{fig1}, we output RDP curves given by our method under different perception parameter $P$ and compare them to the results with known theoretical expressions. The results obtained by our method match well the analytic expression in both scenarios.  Furthermore, we can also plot the results where the analytical solution is not known in Fig \ref{fig1} (b) and (d). 
 We also output the 3D diagram of RDP surface in Fig. \ref{fig2}. For the Gaussian source, we set $S = 4, \sigma = 1$ for visual effect. The results are in accord with those derived from data-driven methods in \cite{blau2019rethinking}. 
\begin{figure}[t]
\centering
\includegraphics[width=\linewidth]{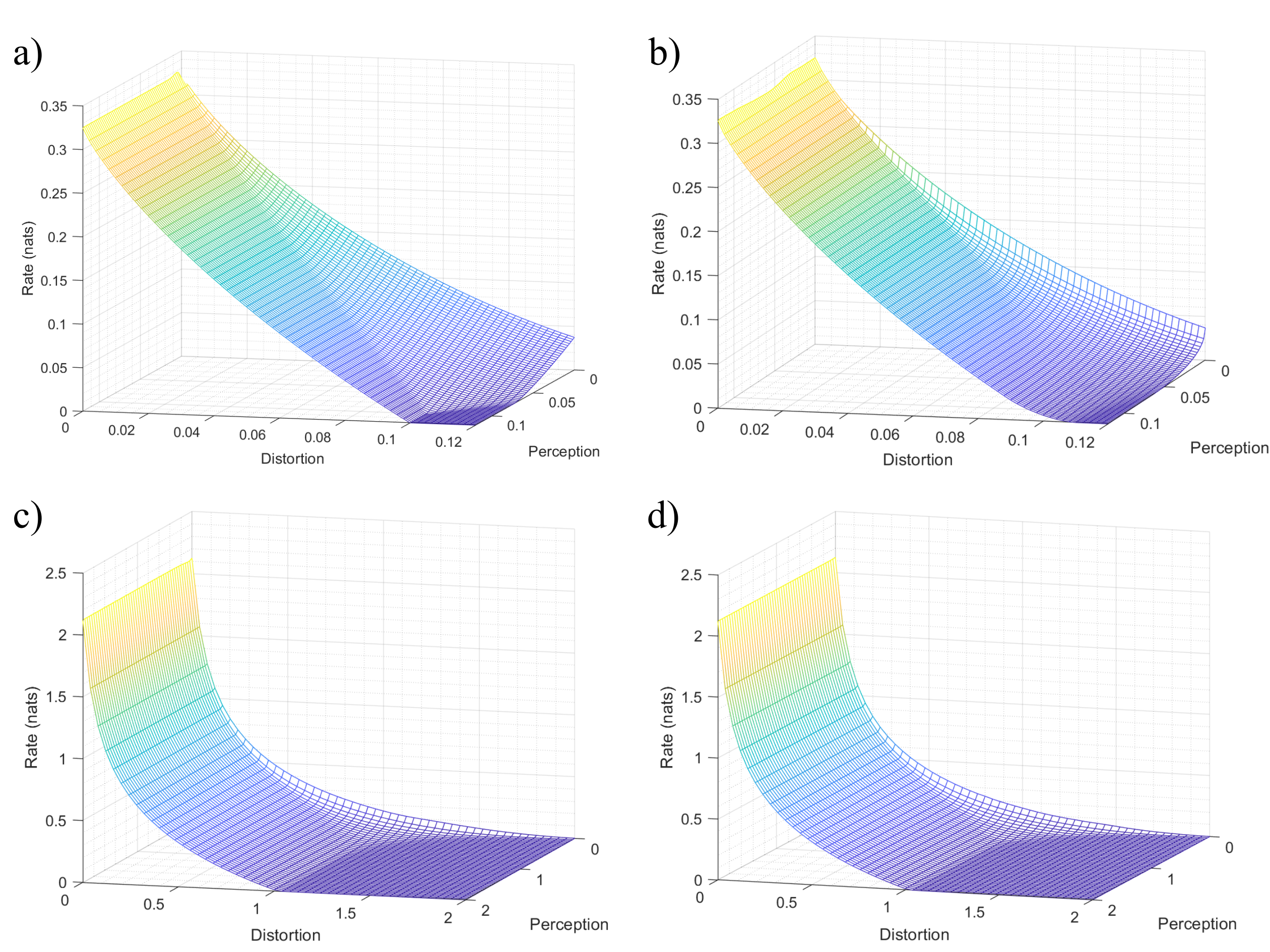}
\caption{The rate-distortion-perception functions obtained by our method. a) the
binary source with TV distance. b) the binary source with KL
divergence. c) the Gaussian source with Wasserstein-2 metric d) the Gaussian source with TV distance.}  \label{fig2}
\end{figure}

\subsection{Convergence Analysis}

We verify the convergence of the improved AS algorithm in this subsection. Here we consider the residual errors of the Karish-Kuhn-Tucker (KKT) condition of the optimization problem (\ref{problem}) to be the indicator of convergence. We define $L_1$ residual errors $r_{\psi}$ as:
    \begin{equation*}
       r_{\psi}=\sum_{j=1}^N\left|\psi_j \sum_{i=1}^M K_{i j} \phi_i p_i- 1\right|.
\end{equation*}
$r_{\phi}$,$r_{\lambda}$,$r_{\eta}$,$r_{\varphi}$,$r_{\xi}$,$r_{\gamma}$ can be defined similarly. We define the overall residual error $r$ to be the root mean square of the above residual errors, i.e.,
    \begin{equation*}
       r=\sqrt{\frac{1}{7}(r_{\psi}^2+r_{\phi}^2+r_{\lambda}^2+r_{\eta}^2+r_{\varphi}^2+r_{\xi}^2+r_{\gamma}^2)}.
\end{equation*}

In Fig. \ref{fig3}, we respectively output the convergent trajectories of $r$ of the binary source with TV distance and Gaussian source with Wasserstein-2 metric against iteration numbers. For the binary source, we set $P = 0.06$. For the Gaussian source, we set $ P = 2$. The results show different convergence behaviors with different distortion parameters. In all of these scenarios, the convergence is eventually achieved below $1 \times 10^{-10}$, indicating that our algorithm exhibits excellent convergence across various situations. 

\begin{figure}[t]
\centering
\includegraphics[width=\linewidth]{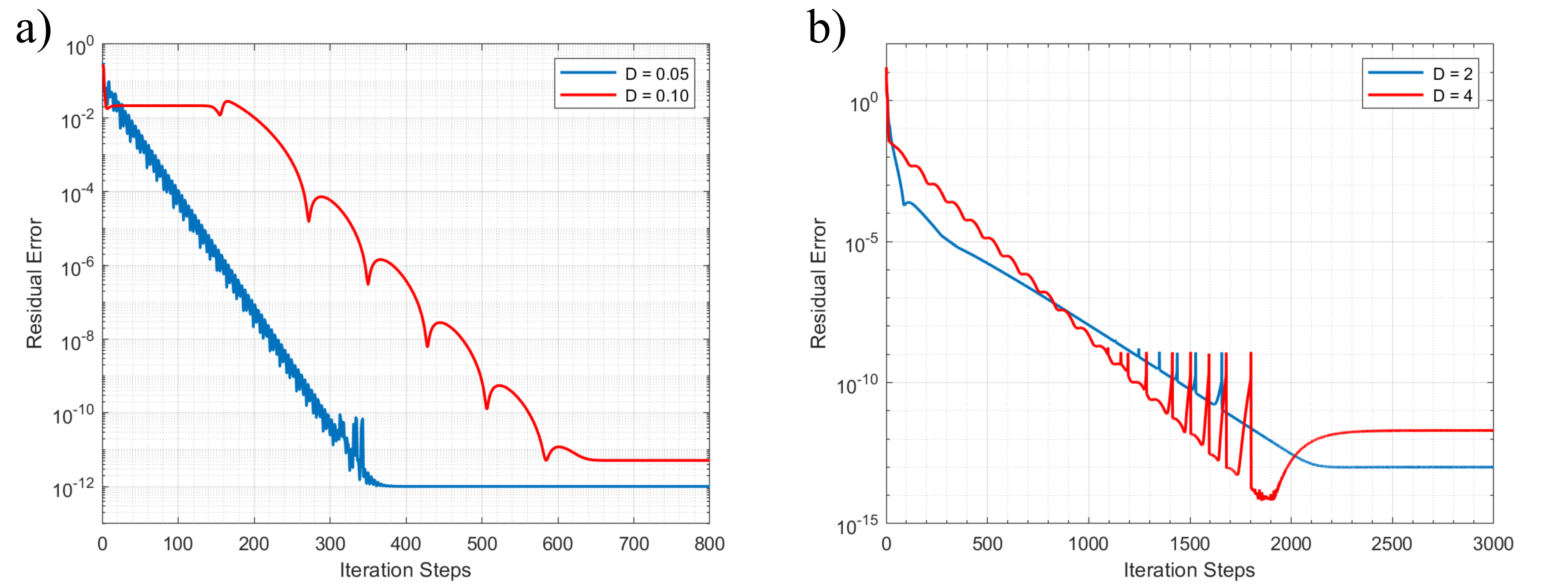}
\caption{The convergent trajectories of $r$ from our method and a) the binary source with different distortion parameters; b) the Gaussian source with different distortion parameters.}  \label{fig3}
\end{figure}

\subsection{Relationship Between Regularization Parameter $\varepsilon$ and Accuracy}

In this subsection, we discuss the effect of regularization parameter $\varepsilon$ on the accuracy of the calculation. We have theoretically demonstrated that when $\varepsilon \rightarrow 0$ the entropy regularized problem converges to the original problem in Theorem \ref{thm-1}. However, we wish to investigate more numerical behavior of $\varepsilon$ to the WBM-RDP.

Results are shown in Table \ref{Tabel2}, where the error represents the $L_1$ difference between the algorithm results and the explicit results for the optimal value. When $\varepsilon$ decreases, the results are more accurate regardless of the source and perception $P$. Furthermore, according to the numerical simulations, when $\varepsilon$ is relatively small, further reducing $\varepsilon$ has limited impact on the error. Therefore, $\varepsilon = 0.01$ seems to be an ideal choice for computing RDP functions, as it ensures a certain level of accuracy.

\begin{table}[t]
\centering
\caption{The computational error against $\varepsilon$ with different sources.}
\label{Tabel2}

\resizebox{\linewidth}{!}
{
\begin{tabular}{c|c|c|c}
\hline
                                 & {$\varepsilon$} &\tiny{$P = 0.02$} &\tiny $P = 0.06$ \\ \hline
\multirow{4}{*}{Binary source}  & \tiny  $1.00 \times 10^{-1}$                       & \tiny $1.26\times 10^{-4}$        & \tiny $4.36\times 10^{-4}$        \\
                                 & \tiny$5.00 \times 10^{-2}$                       & \tiny   $2.84\times 10^{-5}$      & \tiny  $8.75\times 10^{-5}$       \\
                                 & \tiny $1.00 \times 10^{-2}$                       &  \tiny $3.74\times 10^{-6}$       &  \tiny $5.41\times 10^{-6}$        \\
                                 & \tiny $5.00 \times 10^{-3}$                       & \tiny $3.39\times 10^{-6}$        &  \tiny $3.47\times 10^{-6}$        \\
                                  \hline
                                 & {$\varepsilon$}  &\tiny $P = 0.2$    &\tiny $P = 2.0$    \\ \hline
\multirow{4}{*}{Gaussian source} &\tiny $1.00 \times 10^{-1}$                        & \tiny $2.50 \times 10^{-2}$        &\tiny   $1.49 \times 10^{-2}$        \\
                                 &\tiny $5.00 \times 10^{-2}$                        &\tiny  $1.59 \times 10^{-2}$        &\tiny  $7.30 \times 10^{-3}$         \\
                                 &\tiny $1.00 \times 10^{-2}$                        &\tiny  $5.30 \times 10^{-3}$        &\tiny  $2.30 \times 10^{-3}$        \\
                                 &\tiny $5.00 \times 10^{-3}$                        &\tiny $4.20 \times 10^{-3}$         &\tiny   $1.80 \times 10^{-3}$        \\
                                 \hline
\end{tabular}}

\end{table}

\subsection{Distortion-Perception Trade-Off}

\begin{figure}[t]
\centering
\includegraphics[width=\linewidth]{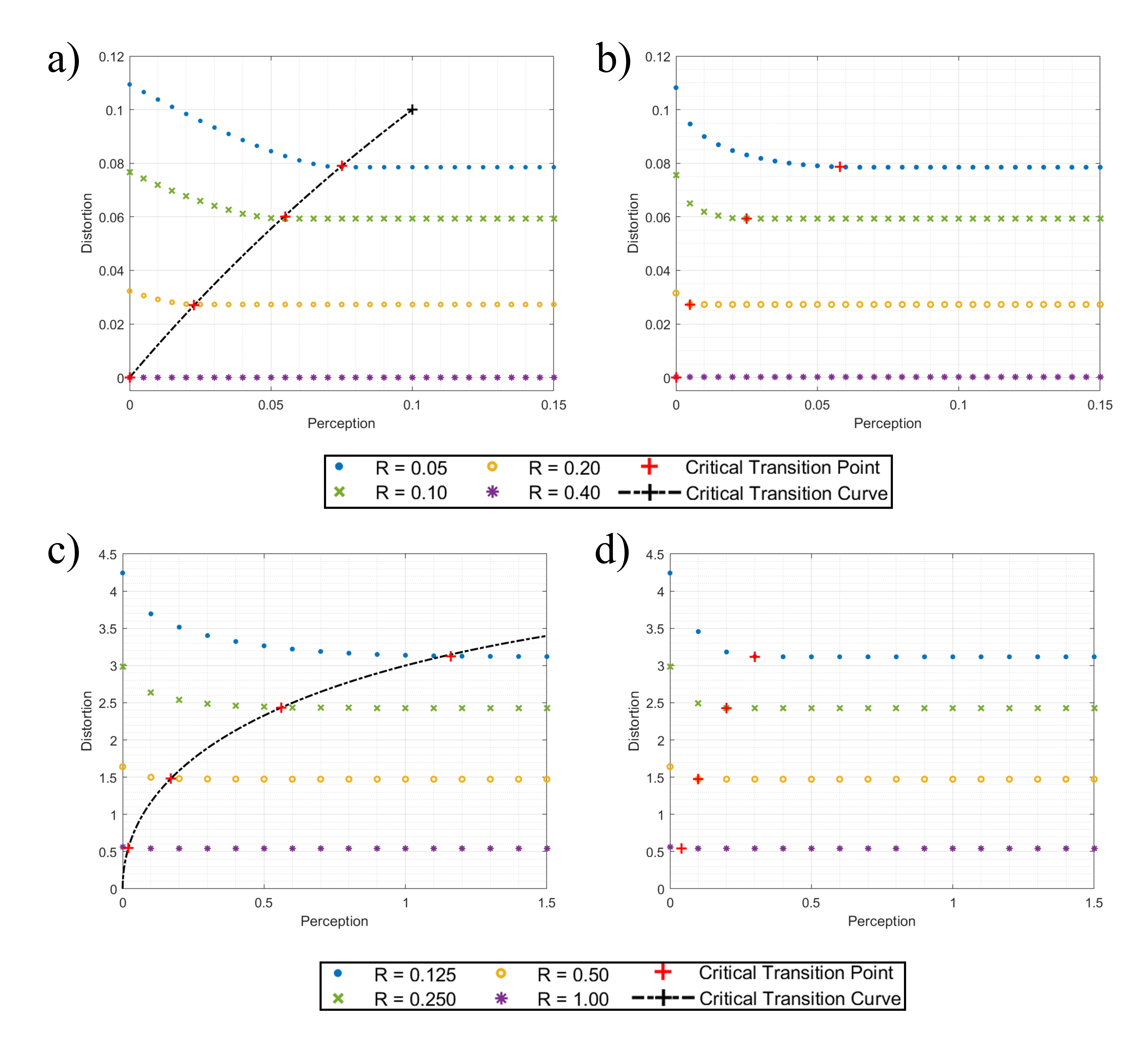}
\caption{Distortion-perception cross-sections under different rates obtained by our method and a) the binary source with TV distance and its analytical critical transition curve; b) the binary source with KL divergence; c) the Gaussian source with Wasserstein-2 metric and its analytical critical transition curve;  d) the Gaussian source with TV distance. Red cross points are the critical transition points obtained by our method.}  \label{fig4}
\end{figure}

\begin{figure}[th]
\centering
\includegraphics[width=\linewidth]{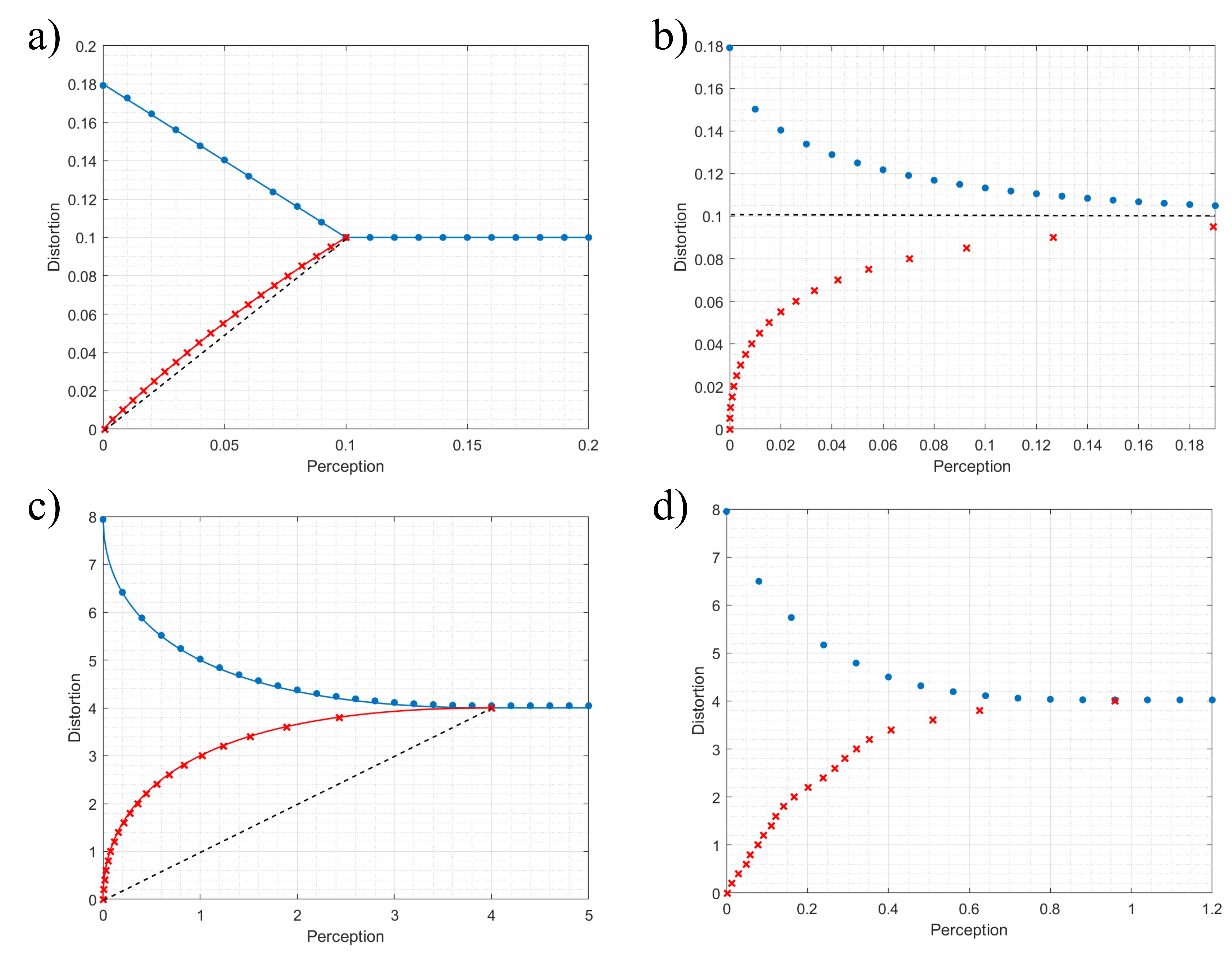}
\caption{Critical transition curve (red) and upper bound (blue) obtained by our method and a) the binary source with TV distance, their analytical solution (corresponding solid line) and auxiliary line $D = P$ (black dashed line); b) the binary source with KL divergence and auxiliary line $P = 0.1$; c) the Gaussian source with Wasserstein-2 metric, their analytical solution and auxiliary line $D = P$;  d) the Gaussian source with TV distance.}  \label{fig5}
\end{figure}
We discuss the distortion-perception trade off with a given rate threshold and the critical transition curve in this part. By exploring this trade-off relationship and calculating the critical transition curve, we can achieve theoretically optimal perception without altering the rate and distortion specified by the original RD functions. This provides a visual perceptual benchmark for potential future compression coding schemes.

In determining the transition curve, we utilize two methods discussed in the text. First, we use the improved AS algorithm to compute the DPR function, yielding a curve illustrating the relationship between distortion and perception index at a specific rate. This analysis enables us to pinpoint the corresponding transition points discussed in Section \ref{sec:interplay} where distortion no longer varies with perception index. Second, to determine the transition boundaries, we can employ the numerical method in \cite{wu2022communication} for solving the RD function to obtain the corresponding $\bm{r}$ given distortion parameter $D$, and then compute $d(\bm{p},\bm{r})$ to obtain the corresponding $P$. To validate the accuracy of our methods, we compare our results with cases having analytical solutions. Once having the analytical expressions of the RDP function, we can obtain its analytical critical transition curve and corresponding upper bound. For binary source with Hamming distortion and TV distance perception, the critical transition curve and upper bound can be written as:
\begin{subequations}
\begin{align*}
    P &= f(D) = \frac{D(1-2p)}{1-2D}, \quad D \in [0,p],\\
    D &= h(P) = \begin{cases}2p(1-p)-(1-2p)P, & P \in[0, p) \\ p, & P \in[p, \infty)\end{cases},
\end{align*}
\end{subequations}
for Gaussian source with squared error distortion and Wasserstein-2 metric perception, the critical transition curve and upper bound can be written as
\begin{subequations}
\begin{align*}
    P &= f(D) = \left(\sigma_X-\sqrt{\left|\sigma_X^2-D\right|}\right)^2, \quad D \in [0,\sigma_X^{2}],\\
    D &= h(P) = \begin{cases}\sigma_X^2+(\sigma_X-\sqrt{P})^2, & P \in[0, \sigma_X^2) \\ \sigma_X^2, & P \in[\sigma_X^2, \infty).\end{cases}
\end{align*}
\end{subequations}

In Fig. \ref{fig4}, we output the distortion-perception cross-sections under different rates $R$ by our method for solving DRP functions. These critical transition points are determined by evaluating their slopes obtained by our method. Specifically, we calculate the absolute values of the slopes of the secant lines between adjacent points to check whether they fall below the threshold of $10^{-3}$, and identify the first point that meets this criterion. In cases where analytical solutions for the transition curve exist, we also draw the corresponding lines to validate the effectiveness of our methods. From this figure, we can observe that for each curve, the critical transition points are consistent with the results obtained from analytical solutions.

To further elucidate the conclusions drawn in Section \ref{sec:interplay} regarding the interplay between distortion and perception as demonstrated in Fig. \ref{curve}, we obtain the corresponding upper bound and critical transition curve in Fig. \ref{fig5} through the solution of $D(0, P)$ and $f(D)$. 
To validate the efficacy of our approach, we have also included analytical results in this context. 
Note that $\text{KL}(\bm{p}\| \bm{r}) \rightarrow +\infty$ when $D \rightarrow D_{\infty}(0)$. The numerical results match well the analytical solutions and the upper bounds in different scenarios. Furthermore, the behavior of $f(D)$ in Fig. \ref{fig5} a) and c) demonstrates well the properties given in Proposition \ref{thm-4} where the critical transition curves are always above dashed line $D = P$ and at the endpoints of both curves distortion index is equal to perception index.

\section{Application to Steganography}
In this section, we apply the above approach of the RDP function to steganographic communication, specifically, to reverse data hiding. The ever-increasing demands for machine-generated content (e.g., picture, video) and digital arts with NFT have motivated the introduction of digital watermarking -- a practice of embedding a message imperceptibly into the content of interest \cite{volkhonskiy2020steganographic}. As an example, the goal of image steganography is to hide a piece of information inside a cover image in a way that is not visible to the human eyes \cite{9335027}. Recently, much attention has been on distribution-preserving steganography using generative models \cite{chen2021distribution}. Traditional digital watermarking methods, exemplified by works such as \cite{1608163}, \cite{tsai2009reversible}, and \cite{6194314}, have historically focused on minimizing the average per-letter distortion following the principles of RD theory. Yet, recent studies have unveiled a critical insight: the sole emphasis on minimizing distortion-based measures, rooted in the RD theory, does not always ensure outcomes with satisfactory perceptual quality \cite{blau2019rethinking}. In the following, we consider image steganography for achieving superior perceptual quality. We introduce the perceptual constraint between the cover image and the stego image into the design. The proposed algorithm provides a guarantee to the imperceptibility, which lies at the heart of steganography.

\begin{figure}[t]
\centering
\includegraphics[width=\linewidth]{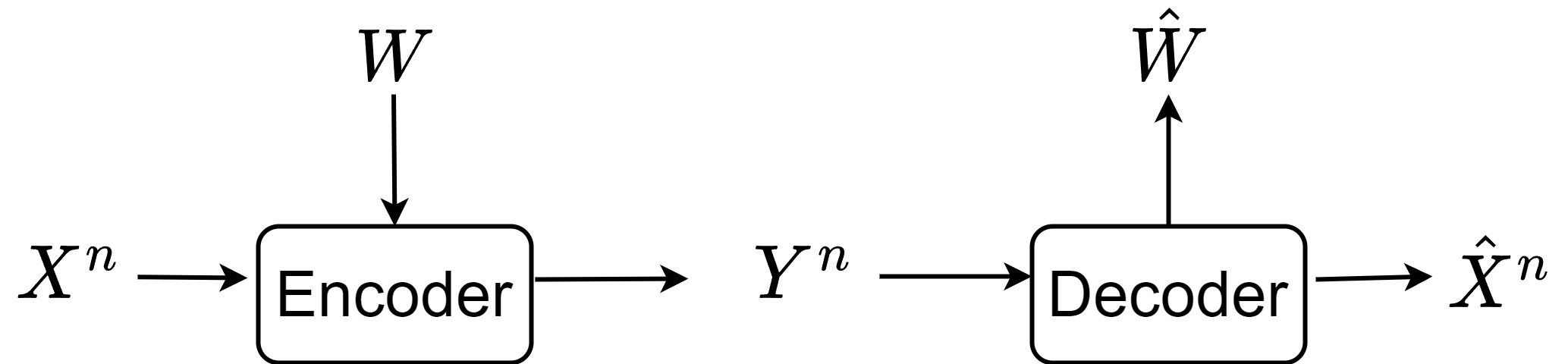}
\caption{Schematic illustration of the steganographic communication -- reversible data hiding system. The host sequence $X^n$ is encoded into marked-sequence $Y^n$ and then decoded into $\hat{X}^n$ such that $X^n= \hat{X}^n$. $W$ is the message to be hidden.}  \label{RDH}
\end{figure}

\subsection{The RDP Function for Reversible Data Hiding}
In the domain of information hiding techniques, reversible data hiding (RDH) stands out due to its unique characteristic, requiring not only precise extraction of the hidden message but also lossless restoration of the cover itself. Mathematically, the message $W = (w_1,\cdots,w_K)$ is embedded by the encoder into the host sequence $X^n = (X_1,\cdots,X_n)$ to produce the marked-sequence $Y^n$ through an injective function $Y^n = f_n(X^n,W)$. We anticipate by elaborate encoding design, the decoder can losslessly reconstruct the message $\hat{W}$ and the host sequence $\hat{X}^n$ through the inverse function $\hat{X}^n = f^{-1}_n(Y^n,W)$ where $X^n= \hat{X}^n$ while keeping $Y^n$ as close as possible to $X^n$. The whole procedure is shown in Fig. \ref{RDH}.
Kalker and Willems \cite{kalker2002capacity} address this question for i.i.d. host sequences. They approach the RDH as a specialized rate-distortion problem and successfully derive the corresponding RD function.

\begin{definition}\label{definition_RDH}
Given a distortion $D\in \mathbb{R}^+$ and source distribution $p_X\in \mathcal{P}(\mathcal{X})$, the RD function for RDH problem is
\begin{equation}\label{eq30}
\begin{aligned}
R_{\text{\emph{RDH}}}(D)= \max _{p_{\hat{X} \mid X}} \quad& H(Y)-H(X)    \\
 \text { \emph{s.t.} }\quad &\mathds{E}[\Delta(X, Y)] \leq D,  
\end{aligned}
\end{equation}%
where $X$ and $Y$ denote the random variables of the host signal and the marked-signal respectively.
\end{definition}

Lin \emph{et al.} put forth a novel scalar code construction aiming at approaching the rate-distortion bound \cite{6194314}. This coding method involves altering the host sequence based on the optimal distribution of marked sequence under the RD function \eqref{eq30}. Under the optimal distribution, the message $W$ is embedded within the host sequence by employing the coding scheme which achieves this boundary. Moreover, this encoding process is reversible, allowing for corresponding decoding. For a comprehensive understanding of the coding technique, readers are encouraged to refer to \cite{6194314} for in-depth analysis.  Since the goal of RDH for image is to conceal information within a cover image such that it remains imperceptible to human visual perception, we believe by incorporating the perceptual constraint into the RDH, it will output visually more satisfying result under the same embedding rate. Therefore we obtain the following RDP function for RDH.
\begin{definition}\label{definition_RDH_RDP}
Given a distortion $D\in \mathbb{R}^+$, a perceptual fidelity $P\in \mathbb{R}^+$ and source distribution $p_X\in \mathcal{P}(\mathcal{X})$, the RDP function for RDH problem is defined as
\begin{equation}\label{eq31}
\begin{aligned}
R_{\text{\emph{RDH}}}(D,P)= \max _{p_{\hat{X} \mid X}} \quad& H(Y)-H(X)   \\
 \text { \emph{s.t.} }\quad &\mathds{E}[\Delta(X, Y)] \leq D, \\  
 &  d\left(p_X, p_Y\right) \leq P,
\end{aligned}
\end{equation}%
where $X$ and $Y$ denote the random variables of host signal and marked-signal respectively. 
\end{definition}
Under the new optimization problem \eqref{eq31}, the sole adjustment required within the original encoding procedure is the substitution of the process for determining the optimal distribution of the marked sequence, while the rest of the steps can remain unchanged. Our WBM-RDP and improved AS algorithm can be directly applied to numerically solve the aforementioned problem \eqref{eq31}.

\begin{figure}[t]
\centering
\includegraphics[width=\linewidth]{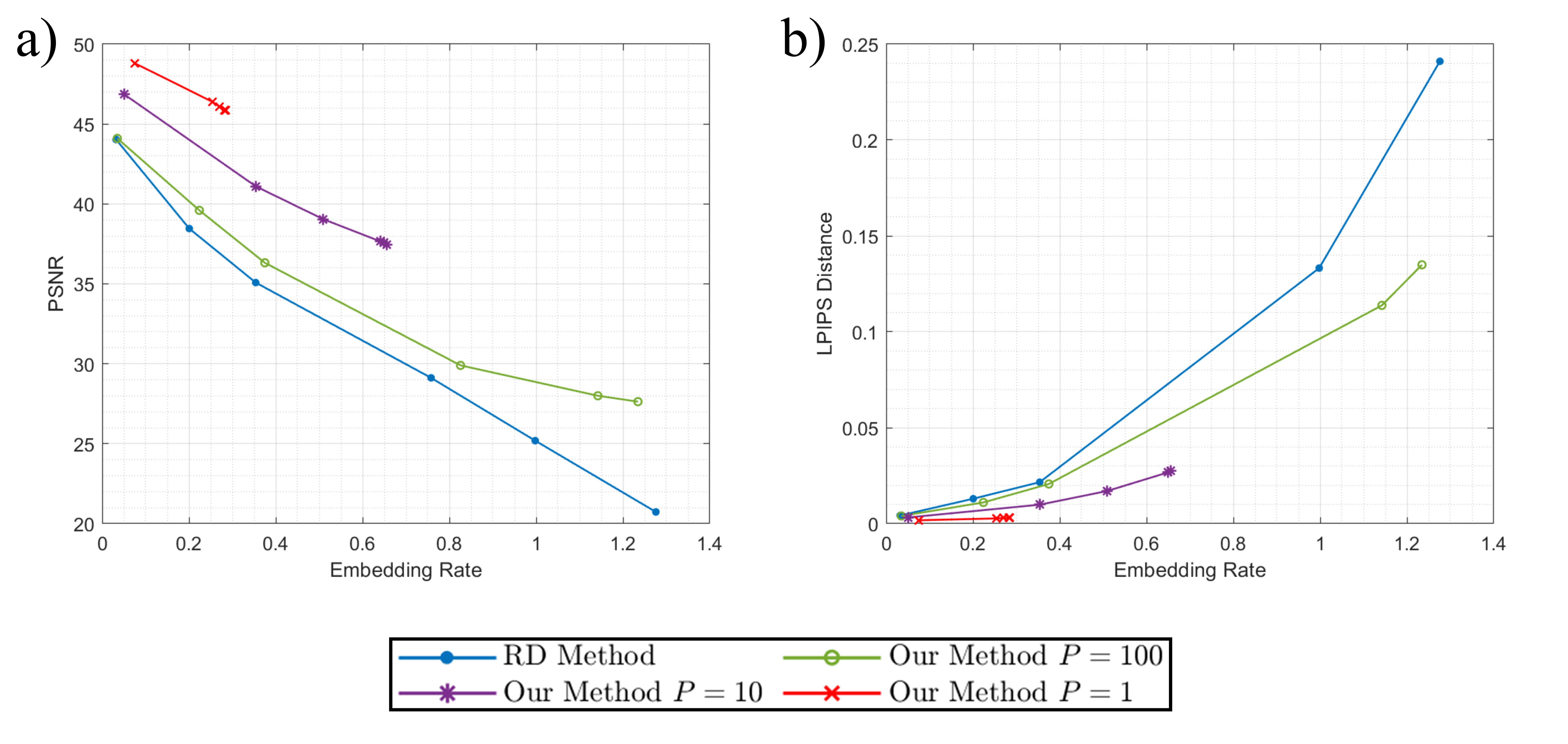}
\caption{Comparison of reconstruction quality with respect to different embedding rates using our method and the baseline methods a) PSNR; b) LPIPS.}  \label{fig_RDH_1}
\end{figure}

\subsection{Experimental Results}
 We applied our algorithm to RDH scheme designed for gray-scale images. Briefly speaking, this approach involves applying the coding method to the prediction errors (PE) of individual pixels. Specifically, for any pixel $h_i$, we use the average value of its neighbour pixels $\hat{h}_i$ to predict it and the PE is defined as
 \begin{equation*}
     \hat{d}_i = (h_i - \hat{h}_i + 128) \ mod \ 256.
 \end{equation*}
 By utilizing the PE as cover signals, we employ the coding method introduced in \cite{6194314} to embed the message into the sequence of PEs, generating a marked-sequence $\bm{y}$. Subsequently, the corresponding pixel of the watermarked image is defined as
 \begin{equation*}
     z_i = (y_i + \hat{h}_i - 128) \ mod \ 256.
 \end{equation*}
 Since the coding method requires the optimal distribution of marked-signals, here we can replace it with the outcome of our Improved AS Algorithm for RDH problem. We compare this new scheme with the method where the perception index is not involved and with different perception parameter $P$ using the test image in Fig.~\ref{fig_RDH_2} and \ref{fig_RDH_3}.

 \begin{figure}[t]
\centering
\includegraphics[width=0.8\linewidth]{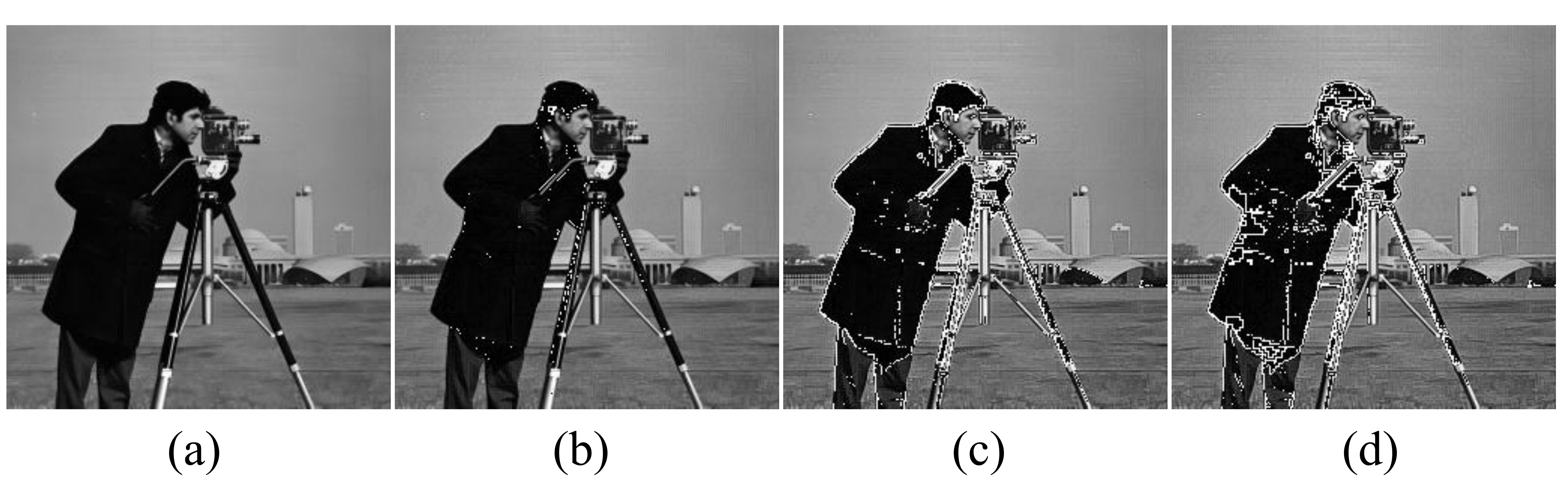}
\caption{Image with increasing perception parameter $P$ when distortion parameter $D = 100$, a) $P = 1$; b) $P = 10$; c) $P = 100$; d) $P = 200$.}  \label{fig_RDH_2}
\end{figure}

\begin{figure}[t]
\centering
\includegraphics[width=0.8\linewidth]{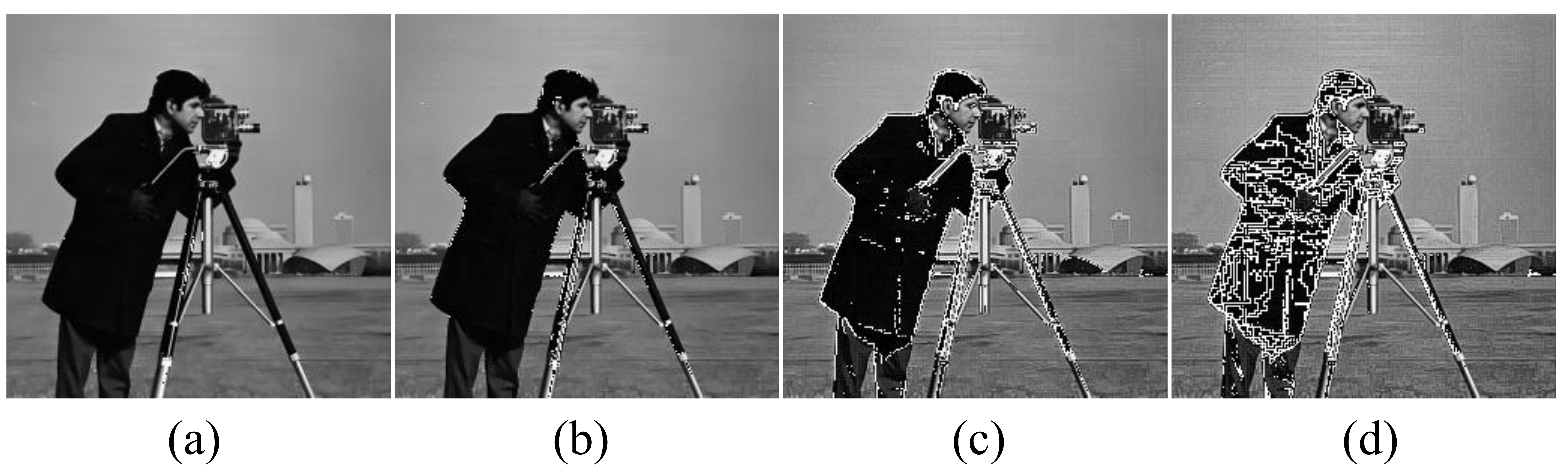}
\caption{Image with increasing distortion parameter $D$ when perception parameter $P = 100$, a) $D = 1$; b) $D = 10$; c) $D = 100$; d) $D = 200$.}  \label{fig_RDH_3}
\end{figure}

The motivation behind introducing the perceptual constraint $P$ into the existing baseline method is to further enhance the quality of images under the same embedding rate, rendering them imperceptible to the human eye. Fig. \ref{fig_RDH_1} a) shows PSNR of the embedded figure obtained by our method and the baseline method. To further elucidate the improvement in image quality brought about by our method, we employ the more perceptually relevant measure LPIPS. Lower LPIPS value indicates higher visual similarity between two images, which is in better accordance with human perception. The simulation result is shown in Fig. \ref{fig_RDH_1} b). Both a) and b) in Fig.~\ref{fig_RDH_1} illustrate that our proposed method effectively enhances the image qualities across various embedding rates, with more significant improvements observed at lower $P$ values. A noteworthy observation is that when incorporating the perception constraint, the embedding rate appears to converge at a certain point 
  in both figures, with a lower convergence embedding rate observed for lower $P$ values. The endpoint corresponds to where the distortion constraint becomes inactive, namely, where the RDP function turns into $R^{\infty}_{\text{{RDH}}}(P),$ and it is monotone increasing against $P$. Therefore, when the amount of information needed for the embedding is relatively small, i.e., when the embedding rate is low, we can opt for a smaller $P$ to meet the corresponding requirements. In this case, the impact on image quality is substantial. Conversely, when a higher embedding rate is required, we should select a larger $P$, although there might be a slight decline in image quality, it still signifies a substantial enhancement compared to the baseline method. Fig. \ref{fig_RDH_2} and Fig. \ref{fig_RDH_3} show the outcomes with different distortion and perception fidelities. Both distortion and perception constraints play a vital role in the RDH, and the introduction of the perceptual constraint yields visually satisfying outcomes while maintaining high embedding rates.  

\section{Conclusion}
We study the information rate-distortion-perception function and its variants. Using techniques in optimal transport, we convert the original problem to a Wasserstein Barycenter model for Rate-Distortion-Perception functions (WBM-RDP). Through our investigation, we explore the interplay between the distortion and perception constraints intrinsic to the RDP functions, including the critical transitions when one of these constraints becomes inactive. Furthermore, we present a novel scheme for computing the information RDP functions by proposing an improved Alternating Sinkhorn method with entropy regularization. Numerical experiments show that our algorithm performs with high accuracy and efficiency. Moreover, we extend our algorithm to an application in steganography, and the outcomes give visually satisfying results in comparison to traditional reversible data hiding methods.

\section*{Acknowledgment}
This work was supported by National Natural Science Foundation of China Grant No. 12271289.

\appendices
\section{Basic Properties of RDP Functions Variants} \label{l1}

\begin{lemma}\label{lemma1}
\begin{align*}
	D_{\infty}(R) &= D(R,+\infty) \leq D(R,P) \leq D(0,P),\quad \forall P \in \mathbb{R}^{+}; \\ P^{\infty}(R)  &= P(R,+\infty) \leq P(R,D) \leq P(0,D),\quad \forall D \in \mathbb{R}^{+}.
\end{align*}
\end{lemma}
\begin{proof}
    The optimal solution to $D(R,P)$ must be the feasible solution to $D_{\infty}(R)$ and also holds true between $P^{\infty}(R) $ and $P(R,D)$, thus the above inequality relations can be obtained from the optimality of $D_{\infty}(R)$ and $P^{\infty}(R) $. The second inequality relations can be proved using a similar technique.
\end{proof}

\begin{lemma}\label{lemma2}
    \begin{align*}
    &P^{\infty}(R)  = 0,\quad \forall R \in \mathbb{R}^{+}\\ &R^{\infty}(P) = 0,\quad \forall P \in \mathbb{R}^{+}.
    \end{align*}
\end{lemma}
\begin{proof}
    Under the original premise, the discrete form of problem \eqref{eq8_2} can be written as:
    \begin{subequations} \label{eq9}
    \begin{align}
    \min _{\bm{w},\bm{r}} \quad  &d(\bm{p},\bm{r})   \\
    \text { s.t. }\quad  &\sum_{j=1}^N w_{i j}=1,\   \sum_{i=1}^M w_{i j} p_i=r_j,  \  \forall i, j,  \label{eq9_b}\\
    &\sum_{i=1}^M \sum_{j=1}^N\left(w_{i j} p_i\right)\left[\log w_{i j}-\log r_j\right] \leq R,\   \sum_{j=1}^N r_j=1 .  
    \end{align}
    \end{subequations}
    We first introduce the zero-padding technique: 
    
    \emph{
    When $M \neq N$, without loss of generality we suppose $M>N,$ we can always set $\bm{\hat{r}}=[\bm{r},\bm{0}_{M\times (M-N)}]$ and $\bm{\hat{w}}=[\bm{w}, \bm{0}_{M\times (M-N)}],$ and the new optimization problem is equivalent to the original optimization problem, since the optimal solutions to both problems satisfy $\bm{\hat{r}^*}=[\bm{r}^*,\bm{0}_{M\times (M-N)}], \bm{\hat{w}^*}=[\bm{w}^*, \bm{0}_{M\times (M-N)}]$ and their objective functions are equal. }
    
    Using the zero-padding technique, we can WLOG set $M = N,$ which will facilitate the following proofs.  For problem \eqref{eq9}, note that for any $R \geq 0$, we can set $w_{ij} = r_j, \bm{r} =\bm{p}$. Therefore, the objective function can always take the minimum value $0$, i.e.,  $P^{\infty}(R)  = 0, \forall R \in \mathbb{R}^{+}$.
    
   Similarly, $R^{\infty}(P)$ can be written as:
    \begin{equation} \label{eq9_0}
    \begin{aligned}
    \min _{\bm{w},\bm{r}} \quad  &\sum_{i=1}^M \sum_{j=1}^N\left(w_{i j} p_i\right)\left[\log w_{i j}-\log r_j\right]   \\
    \text { s.t. }\quad  &\sum_{j=1}^N w_{i j}=1,\   \sum_{i=1}^M w_{i j} p_i=r_j,  \  \forall i, j,  \\
    &d(\bm{p},\bm{r}) \leq P,\   \sum_{j=1}^N r_j=1 .  
    \end{aligned}
    \end{equation}
    Under any $P \geq 0$, we can still set $w_{ij} = r_j, \bm{r} =\bm{p}$ and the objective function can always take the minimum value $0$, i.e.,  $R^{\infty}(P) = 0, \forall P \in \mathbb{R}^{+}$. 
\end{proof}
In \cite{4069146} and \cite{cuff2010coordination}, it is shown that with common randomness, empirical coordination can be achieved with zero rate, which nicely substantiates the lemma.

\begin{lemma}\label{lemma3}
Given $R_0$, if $D_{\infty}(R_0) = 0$, then $$D(R_0,P) = 0, \forall P \in \mathbb{R}^{+}.$$
\end{lemma}
\begin{proof}
    The discrete form of \eqref{eq8_3} can be written as:
        \begin{subequations} \label{eq10}
        \begin{align}
    \min _{\bm{w},\bm{r}} \quad  &\sum_{i=1}^M \sum_{j=1}^N w_{i j} p_i d_{i j}  \label{eq10_a}\\
    \text { s.t. } \quad &\sum_{j=1}^N w_{i j}=1,\   \sum_{i=1}^M w_{i j} p_i=r_j,  \  \forall i, j,  \label{eq10_b}\\
    &\sum_{i=1}^M \sum_{j=1}^N\left(w_{i j} p_i\right)\left[\log w_{i j}-\log r_j\right]  \leq R,\   \sum_{j=1}^N r_j=1 ,  \label{eq10_c}
        \end{align}
        \end{subequations}
        if $D_{\infty}(R_0) = 0$, according to the property of distortion matrix $d_{ij}$, the non-diagonal elements of $p_iw_{ij}$ must be zeros. With the constraint \eqref{eq10_b} we can derive $\bm{r} = \bm{p}$, thus $d(\bm{p},\bm{r}) = 0$ and $D(R_0,P) = 0, \forall P \in \mathbb{R}^{+}$.
\end{proof}
This lemma tells us that if zero distortion is achievable under a given rate, then any perceptual quality can also be achieved.

\begin{lemma}\label{lemma3.5} 
    When $D = 0$, we have
    \begin{align*}
    R(D, P) &= R_{\infty}(D),\quad \forall P \in \mathbb{R}^{+};\\
    P(R, D) &= 0, \quad\forall R \in \mathbb{R}^{+}.
    \end{align*}
\end{lemma}
\begin{proof}
    When $D=0$ we can derive $\bm{r} = \bm{p}$ according to the property of distortion matrix like Lemma \ref{lemma3}, thus the perception constraint is inactive and $R(0, P)= R^{\infty}(0), \forall P \in \mathbb{R}^{+}$. For the second equality, we can set $w_{ij} = r_j$ and the rate constraint can be always satisfied, i.e., $P(R, 0) = 0, \forall R \in \mathbb{R}^{+}$.
\end{proof}

\begin{lemma}\label{lemma4}
    Both $D(R,P)$ and $P(R,D)$ are convex and non-increasing with respect to corresponding independent variables.
\end{lemma}
\begin{proof}
    Similar to $R(D,P)$, since all functions in the optimization problems are convex with respect to conditional distributions $p_{\hat{X} \mid X}$, both $D(R,P)$ and $P(R,D)$ are convex. The non-increasingness is due to the fact that both problems are the minimum over increasingly larger sets as corresponding independent variables increase.
\end{proof}

\section{Proof of Theorem~\ref{thm-0}}\label{app:proof-thm0}
\begin{proof}
   Suppose the optimal solution to (\ref{eq0_1}) is $\{ \bm{w}^*,\bm{r}^*,\bm{\Pi}^* \}$. 
    First, according to the property of the Wasserstein metric and the constraint, we can get 
    \begin{equation*}
        W(\bm{p},\bm{r^*}) \leq \sum_{i=1}^M \sum_{j=1}^N \Pi^*_{i j} c_{i j} \leq D,
    \end{equation*}    
    thus $\{ \bm{w}^*,\bm{r}^* \}$ are the feasible solution to (\ref{eq0_0}). 
    Let $\{ \hat{\bm{w}},\hat{\bm{r}} \}$ be the optimal solution to (\ref{eq0_0}), and $\hat{\bm{\Pi}}$ be the optimal transport plan of the Wasserstein metric between $\{ \hat{\bm{w}},\hat{\bm{r}} \}$, we denote 
     \begin{equation*}
    L(\bm{w},\bm{r})= \sum_{i=1}^M \sum_{j=1}^N\left(w_{i j} p_i\right)\left[\log w_{i j}-\log r_j\right] .
    \end{equation*}
    If $L(\hat{\bm{w}},\hat{\bm{r}}) < L(\bm{w}^*,\bm{r}^*)$, then $\{ \hat{\bm{w}},\hat{\bm{r}}, \hat{\bm{\Pi}} \}$ is the feasible solution to (\ref{eq0_1}) whose target value is lower than that of $\{ \bm{w}^*,\bm{r}^*,\bm{\Pi}^*\}$, which leads
    to contradiction. 
    Therefore $\{ \bm{w}^*,\bm{r}^* \}$ is the optimal solution to (\ref{eq0_0}). 
\end{proof}

\section{Proof of Theorem~\ref{thm-3}}\label{app:proof-thm3}
\begin{proof}
        1) Assuming $D(R, P) = D_{\infty}(R)$ when $D > 0$, according to the RD theory \cite{wu2022communication} and Proposition \ref{proposition1}, the optimal value of $R$ is uniquely determined once $D$ is determined. Therefore, we can set $f(D) = P(R_{\infty}(D), D)$, which is continuous since $D_{\infty}(R)$ and $P(R, D)$ are continuous as both functions are convex from Proposition \ref{proposition1} and Lemma \ref{lemma4} (see Appendix A, below are the same) in an open set. Furthermore, the continuity of $D(R,P)$ with respect to $P$ guarantees the convergence in \eqref{eq10_0_1_a}.
        
        When $D = 0$, we can set $P = f(D) = 0$ based on Lemma \ref{lemma3}. As $D$ approaches $0$, the non-diagonal elements of $p_iw_{ij}$ will converge to $0$, resulting in $\bm{r} \rightarrow \bm{p}$ according to the boundary conditions. Therefore, $P \rightarrow 0$, and $f(D)$ is continuous on $\mathbb{R}^{+}$. 
        
        According to Lemma \ref{lemma1}, when $P > f(D)$ we have 
        \begin{equation*}
            D_{\infty}(R) \leq D(R,P) \leq D(R,f(D)) = D_{\infty}(R),
        \end{equation*}
        thus $D(R, P) = D_{\infty}(R)$ holds true when $P \geq f(D)$.
        
        Suppose there exist $D_0$ and $P_0$ such that $P_0 < f(D_0)$  and $D(R, P_0) = D_{\infty}(R)$ remains correct. If $D_0 >0$, $R_{\infty}(D_0)$ is uniquely determined, and the above statement is equivalent to the existence of one feasible solution $p^{*}_{\hat{X} \mid X}$ to
        \begin{subequations}
        \begin{align*}
        P(R_{\infty}(D_0),D_0) = \min _{p_{\hat{X} \mid X}} \quad&d\left(p_X, p_{\hat{X}}\right)    \\
        \text { s.t. }\quad   &  I(X, \hat{X}) \leq R_{\infty}(D_0),  \\
        &\mathds{E}[\Delta(X, \hat{X})] \leq D_0 , 
        \end{align*}
        \end{subequations}
        and its corresponding objective value satisfies $P_0 < P(R_{\infty}(D_0),D_0)$.  However, the optimality of $P(R, D)$ contradicts this conclusion since $P(R_{\infty}(D_0),D_0)$ is the optimal value among all the feasible solutions. If $D_0 = 0$, $f(D)$ reaches its theory minimal $0$ according to Lemma \ref{lemma3.5}.  Therefore, $f(D)$ is the minimal function among all possible functions that when $P \geq f(D), D(R, P) = D_{\infty}(R)$.

         $D$ and $R$ is bijective relation when $D\leq D_{\infty}(0)$ according to the RD theory, thus $D(R, P) = D_{\infty}(R)$ also means $R(D, P) = R_{\infty}(D)$. We can obtain $R(D_{\infty}(0),P) = 0$ when $P = f(D_{\infty}(0))$ since the optimal solution to $f(D_{\infty}(0))$ is also the feasible solution to $R(D_{\infty}(0),f(D_{\infty}(0)))$ and reach theoretical lower bound $0$. Furthermore, the continuity of $R(D,P)$ with respect to $P$ guarantees the convergence in \eqref{eq10_0_1_b}.
         
         When $D > D_{\infty}(0)$, we have $R_{\infty}(D) = 0$. From above we can obtain $R(D_{\infty}(0),P) = 0$ when $P = f(D_{\infty}(0))$, thus when $P \geq f(D_{\infty}(0))$ and $D > D_{\infty}(0)$ we have:
         \begin{equation*}
            0 \leq R(D,P) \leq R(D_{\infty}(0),f(D_{\infty}(0))) = 0,
        \end{equation*}
         the second inequality is due to the non-increasing of $R(D, P)$ mentioned in Proposition \ref{proposition1}. Therefore $R_{\infty}(D) = R(D,P)$ when $P \geq f(D)$ still holds. Furthermore, the statement that $f(D)$ is the minimal function can also be derived from the optimality of $P(R,D)$. 
         
        2) According to Lemma \ref{lemma2}, $P^{\infty}(R)  = 0, \ \forall R$. The optimal solution to $D(0,0)$ must be the feasible solution to $P(0,D(0,0))$ with the objective value being 0, thus $P(0,D(0,0)) = 0$ as the non-negativity of the perception measure under Assumption \ref{assumption-perception}. For any $R$ and $D \geq D(0,0)$ we have
        \begin{equation*}
            0 \leq P(R,D) \leq P(0,D(0,0)) = 0,
        \end{equation*}
        thus $P(R,D) = P^{\infty}(R)  = 0, \quad R \in \mathbb{R}^+$.

        3)  According to Lemma \ref{lemma2}, for any $P \in \mathbb{R}^{+}$ we have $R^{\infty}(P) = 0$. If we want $R(D,P) = R^{\infty}(P)$, we must find the root value $D$ of $R(D,P)$ when given $P$ and this problem can be converted to
        \begin{equation} \label{eq10_0}
        \begin{aligned}
        h(P) = \min _{\bm{r}} \quad  &\sum_{i=1}^M \sum_{j=1}^N r_{j} p_i d_{i j}  \\
        \text { s.t. }\quad& d(\bm{p},\bm{r}) \leq P,\   \sum_{j=1}^N r_j=1,   
        \end{aligned}
        \end{equation}
        since $h(P)$ is convex on $\mathbb{R}^{+}$ under Assumption \ref{assumption-convex}, thus $h(P)$ is continuous on $\mathbb{R}^{++}$. When $P = 0$, $\bm{r}=\bm{p}$. With $P \rightarrow 0$, $\bm{r} \rightarrow \bm{p}$ for the property of the perception measure under Assumption \ref{assumption-perception}, and since the objective function is continuous with respect to $\bm{r}$, thus $h(P) \rightarrow h(0)$ and $h(P)$ is continuous on $\mathbb{R}^{+}$. Furthermore, the continuity of $R(D,P)$ with respect to $D$ guarantees the convergence in \eqref{eq10_0_2}. Similarly, the optimality of \eqref{eq10_0} can also prove $h(P)$ is the minimum of all the possible functions that when $D \geq h(P), R(D,P) = R^{\infty}(P)$.
\end{proof}

\section{Proof of Proposition~\ref{thm-4}}\label{app:proof-thm4}
\begin{proof}
When $R = 0$ in \eqref{eq10}, then $w_{ij} = r_j$ according to \eqref{eq10_c}, and the model is simplified to
        \begin{equation} \label{eq11}
        \begin{aligned}
    \min _{\bm{r}} \quad  &\sum_{i=1}^M \sum_{j=1}^N r_{j} p_i d_{i j}  \\
    \text { s.t. }\quad& \sum_{j=1}^N r_j=1 ,  
        \end{aligned}
        \end{equation}
due to the non-negativity of this optimization, we can get:
\begin{subequations}
\begin{align*}
    r_j =& \begin{cases}1 & \text {if} \quad j = \arg\min_{j}\sum_{i=1}^M p_i d_{i j}  \\ 0 & \text {else}\end{cases},\\
    D =& \quad \min_{j}\sum_{i=1}^M p_i d_{i j},
\end{align*}
\end{subequations}
given the rate $R = 0$ and distortion $D_{\infty}(0)$ threshold, the optimization model for perception can be simplified as 
\begin{subequations} \label{eq12}
\begin{align}
   \min _{\bm{\Pi}} \quad  &\sum_{i=1}^M \sum_{j=1}^N \Pi_{i j} c_{i j}  \\
    { \text { s.t. } } \quad
    &\sum_{i=1}^M \Pi_{i j}= r_j ,\  \sum_{j=1}^N \Pi_{i j}= p_i, \  \forall i, j, \label{eq12_b}
\end{align}
\end{subequations}
here we use the equivalence proved in Theorem \ref{thm-0}. Since $\bm{r}$ is the one-hot vector, the optimization problem is easy to solve:
\begin{subequations}
\begin{align*}
    \Pi_{ij} =& \begin{cases} p_i & \text {if} \quad j = \arg\min_{j}\sum_{i=1}^M p_i d_{i j}  \\ 0 & \text {else}\end{cases},\\
    P =& \quad \sum_{i=1}^M p_i c_{i j}, \quad j = \arg\min_{j}\sum_{i=1}^M p_i d_{i j},
\end{align*}
\end{subequations}
when $c_{ij} = d_{ij}$, $D_{\infty}(0) = P(0,D_{\infty}(0))$. For $\forall D \in [0,D_{\infty}(0)]$, we have
\begin{equation*}
\sum_{i=1}^M \sum_{j=1}^N w_{i j} p_i d_{i j} \leq D.
\end{equation*}
 Notice matrix $w_{i j} p_i$ can always satisfy \eqref{eq12_b} according to \eqref{eq9_b}, thus we have
\begin{equation*}
P \leq \sum_{i=1}^M \sum_{j=1}^N w_{i j} p_i c_{i j}.
\end{equation*}
When $c_{ij} = d_{ij}$, we can get $ D \geq f(D), D \in [0,D_{\infty}(0)]$. 

\end{proof}

\section{RDP Functions Variants with $n$-product Distributions}\label{l2}
\subsection{Properties}
\begin{lemma}\label{lemma5}
\begin{align*}
	D_{\infty}(R) &= D^{[n]}(R,+\infty) \\
 &\leq D^{[n]}(R,P) \leq D^{[n]}(0,P), \forall P \in \mathbb{R}^{+}; \\ P^{\infty^{[n]}}(R) &= P^{[n]}(R,+\infty) \\
 &\leq P^{[n]}(R,D) \leq P^{[n]}(0,D), \forall D \in \mathbb{R}^{+}.
\end{align*}
\end{lemma}

\begin{proof}
    Similar to Lemma \ref{lemma1}, the above inequality relations can be obtained from the optimality of $D_{\infty}^{[n]}(R), P^{\infty^{[n]}}(R)$ and the non-increasing properties of $D^{[n]}(R,P), P^{[n]}(R,D)$.
\end{proof}

\begin{lemma}\label{lemma6}
\begin{align*}
P^{\infty^{[n]}}(R) = 0, \forall R \in \mathbb{R}^{+}; R^{\infty^{[n]}}(P) = 0, \forall P \in \mathbb{R}^{+}.
\end{align*}
\end{lemma}
\begin{proof}
    Note that when $\bm{r} =\bm{p}$, we still have $d\left(p_{X^n}, p_{\hat{X}^n}\right) = 0$. Therefore, we can always set $w_{ij} = r_j, \bm{r} =\bm{p}$ in both cases like Lemma \ref{lemma2} and the corresponding  $P^{\infty^{[n]}}(R), R^{\infty^{[n]}}(P)$ can reach the minimum value $0$.
\end{proof}

\begin{lemma}\label{lemma7}
Given $R_0$, if $D_{\infty}(R_0) = 0$, then $$D^{[n]}(R_0,P) = 0, \forall P \in \mathbb{R}^{+}.$$
\end{lemma}
\begin{proof}
    According to Lemma \ref{lemma3}, if $D_{\infty}(R_0) = 0$ we can derive $\bm{r} = \bm{p}$. Therefore $d\left(p_{X^n}, p_{\hat{X}^n}\right) = 0$ and  $D^{[n]}(R_0,P) = 0, \forall P \in \mathbb{R}^{+}$.
\end{proof}

\begin{lemma}\label{lemma8} 
    When $D = 0$,
    $$ R^{[n]}(0, P)= R_{\infty}(0), \forall P \in \mathbb{R}^{+};  P^{[n]}(R, 0) = 0, \forall R \in \mathbb{R}^{+}.$$
\end{lemma}
\begin{proof}
     According to Lemma \ref{lemma3.5}, if $D = 0$ we can derive $\bm{r} = \bm{p}$. Therefore $d\left(p_{X^n}, p_{\hat{X}^n}\right) = 0$, thus the perception constraint is inactive and $R^{[n]}(0, P)= R^{\infty}(0), \forall P \in \mathbb{R}^{+}$. For the second equality, we can set $w_{ij} = r_j$ and the rate constraint can be always satisfied, i.e., $P^{[n]}(R, 0) = 0, \forall R \in \mathbb{R}^{+}$.
\end{proof}

\subsection{Proof of Proposition~\ref{thm-5}}
\begin{proof}

    1) Similar to Theorem \ref{thm-3}-1, we can set$$f^{[n]}(D) = P^{[n]}(R_{\infty}(D), D).$$
    Suppose there exist $D_1$ and $P_1$ such that $P_1 < f^{[n]}(D_1)$  and \eqref{n1} or \eqref{n2} remains correct. If $D_1 >0$ the optimality of $P^{[n]}(R,D)$ contradicts the above conclusion, and $f^{[n]}(D)$ reaches theoretical minimal $0$ when $D_1 = 0$ according to Lemma \ref{lemma8}, thus the minimality of $f^{[n]}$ is guaranteed.

    2) According to Lemma \ref{lemma6}, $P^{\infty^{[n]}}(R) = 0, \forall R \in \mathbb{R}^{+}$. Using the same technique in Theorem \ref{thm-3}-2, the optimal solution to $D^{[n]}(0,0)$ must be the feasible solution to $P^{[n]}(0,D^{[n]}(0,0))$ with the objective value being 0, thus $P^{[n]}(0,D^{[n]}(0,0)) = 0$. 
    
    For any $R$ and $D \geq D^{[n]}(0,0)$ we have
        \begin{equation*}
            0 \leq P^{[n]}(R,D) \leq P^{[n]}(0,D^{[n]}(0,0)) = 0,
        \end{equation*}
        according to Lemma \ref{lemma5}, thus $$P^{[n]}(R,D) = P^{\infty^{[n]}}(R) = 0, \quad R \in \mathbb{R}^+.$$
    
    3) According to Lemma \ref{lemma6}, for any $P \in \mathbb{R}^{+}$ we have $R^{\infty^{[n]}}(P) = 0$, and we can find the critical value $D$ like Theorem \ref{thm-3}-3 by solving the following problem 
        \begin{equation}
        \begin{aligned}\label{n3}
        h^{[n]}(P) = \min _{\bm{r}} \quad  &\sum_{i=1}^M \sum_{j=1}^N r_{j} p_i d_{i j}  \\
        \text { s.t. }\quad& d\left(p_{X^n}, p_{\hat{X}^n}\right) \leq P,\   \sum_{j=1}^N r_j=1.   
        \end{aligned}
        \end{equation}
    Similarly, the optimality of \eqref{n3} can also prove $h^{[n]}(P)$ is the minimum of all the possible functions that when $D \geq h^{[n]}(P), R^{[n]}(D,P) = R^{\infty^{[n]}}(P)$. 

\end{proof}

\section{Proof of Theorem~\ref{thm-1}}\label{app:proof-thm1}
\begin{proof}
    Following the convergence proof for entropic optimal transport in \cite{nutz2022entropic}, we consider a sequence $\{ \bm{w}_{k},\bm{\Pi}_{k},\bm{r}_{k} \}_k$ such that $\varepsilon_k \rightarrow 0$ and $\varepsilon_k >0$ and $\{\bm{w}_{k},\bm{\Pi}_{k},\bm{r}_{k} \}$ denotes the solution to (\ref{problem}) with $ \varepsilon = \varepsilon_k$.

Since the feasible region of (\ref{problem}) is close and bounded, we can take a converging sub-sequence $ \{ \bm{w}_{k},\bm{\Pi}_{k},\bm{r}_{k} \} \rightarrow \{ \bm{w}^{*},\bm{\Pi}^{*},\bm{r}^{*} \}$. For any optimal solution $\left\{ \bm{w},\bm{\Pi},\bm{r} \right\}$ to (\ref{eq0_1}), according to optimality, we can have
\begin{equation}\label{proof2}
   0 \leq L(\bm{w}_k,\bm{r}_k)-L(\bm{w},\bm{r}) \leq \varepsilon_k \Big(E(\bm{\Pi})-E(\bm{\Pi}_k)\Big).
\end{equation}
Here $L(\bm{w},\bm{r})= \sum_{i=1}^M \sum_{j=1}^N\left(w_{i j} p_i\right)\left[\log w_{i j}-\log r_j\right] $. 


Note that $E(\cdot), L(\cdot,\cdot)$ are continuous, thus taking limit $k \rightarrow \infty$ in (\ref{proof2}), we can get 
\begin{equation*}
   L(\bm{w}^*,\bm{r}^*)=L(\bm{w},\bm{r}).
\end{equation*}
Furthermore, dividing (\ref{proof2}) by $\varepsilon_k$, we can get 
\begin{equation*}
   E(\bm{\Pi}_k) \leq  E(\bm{\Pi}),
\end{equation*}
and taking the limit can also output 
\begin{equation*}
   E(\bm{\Pi}^*) \leq  E(\bm{\Pi}),
\end{equation*} 
 which shows that $\{ \bm{w}^{*},\bm{\Pi}^{*},\bm{r}^{*} \}$ is the optimal solution to problem (\ref{proof1}).

Since $E(\cdot)$ is a strict convex function, optimization problem (\ref{proof1}) has a unique optimal solution and thus the whole sequence $\{ \bm{w}_{k},\bm{\Pi}_{k},\bm{r}_{k} \}_k$ converges to $\{ \bm{w}^{*},\bm{\Pi}^{*},\bm{r}^{*} \}$.     
\end{proof}

\bibliography{ref.bib}

\end{document}